\documentclass[journal]{IEEEtran}
\usepackage{amsmath,amsfonts,amsthm} 
\usepackage{graphicx}
\usepackage[table,pdftex,hyperref,svgnames]{xcolor}
\usepackage{enumitem,xspace}
\usepackage{mathtools}

\usepackage{hyperref,doi}
\hypersetup{colorlinks=true, linkcolor=blue, breaklinks=true, urlcolor=blue, citecolor=blue,}
\usepackage{url}

\usepackage{graphicx}
\usepackage{subcaption}
\usepackage{multicol} 
\usepackage{mathtools,bbm}
\usepackage{multirow}

\usepackage[noend]{algorithmic}
\algsetup{indent=2em}

\usepackage{algorithm}

\newcommand{\sign}{\operatorname{sign}}
\newcommand{\mcU}{\mathcal{U}}
\newcommand{\mcV}{\mathcal{V}}
\newcommand{\mcS}{\mathcal{S}}
\newcommand{\realpositive}{\mathbb{R}_{>0}}
\newcommand{\fLVe}{\subscr{f}{LVe}}

\usepackage{multicol}
\usepackage{version,xspace}
\usepackage{url,doi}
\newcommand{\until}[1]{\{1,\dots, #1\}}

\newcommand{\subscr}[2]{#1_{\textup{#2}}}
 \newcommand{\setdef}[2]{\{#1
  \; | \; #2\}} 
\newcommand{\bigsetdef}[2]{\big\{#1 \; | \; #2\big\}}
\newcommand{\Bigsetdef}[2]{\Big\{#1 \; \big| \; #2\Big\}}
\newcommand{\map}[3]{#1: #2 \rightarrow #3}

\newcommand{\jac}[1]{D\mkern-2.5mu{#1}}

\newcommand\oprocendsymbol{\hbox{$\triangle$}}
\newcommand\oprocend{\relax\ifmmode\else\unskip\hfill\fi\oprocendsymbol}
\def\eqoprocend{\tag*{$\triangle$}}

\usepackage{enumitem} \renewcommand{\theenumi}{(\roman{enumi})}

\usepackage[prependcaption,colorinlistoftodos]{todonotes}

\DeclareSymbolFont{bbold}{U}{bbold}{m}{n}
\DeclareSymbolFontAlphabet{\mathbbold}{bbold}
\newcommand{\vect}[1]{\mathbbold{#1}}

\newcommand{\real}{\mathbb{R}}
\newcommand{\complex}{\mathbb{C}}

\DeclareMathOperator{\diag}{diag}

\DeclareMathOperator{\Ker}{\mathrm{Ker}}
\DeclareMathOperator{\Img}{\mathrm{Img}}
\renewcommand{\top}{\mathsf{T}} 

\newtheorem{theorem}{Theorem}
\newtheorem{proposition}[theorem]{Proposition}

\newtheorem{lemma}[theorem]{Lemma}
\newtheorem*{lemma*}{Lemma}
\theoremstyle{definition}
\newtheorem{definition}[theorem]{Definition}
\newtheorem{remark}[theorem]{Remark}
\newtheorem{example}[theorem]{Example}

\newtheorem*{example*}{Example}

\newcommand{\suchthat}{\;\ifnum\currentgrouptype=16 \middle\fi|\;}
\newcommand{\scirc}{\raise1pt\hbox{$\,\scriptstyle\circ\,$}}

\newcommand{\norm}[1]{\| #1 \|}
\newcommand{\verti}[1]{{\left\vert\kern-0.25ex\left\vert\kern-0.25ex\left\vert #1 
    \right\vert\kern-0.25ex\right\vert\kern-0.25ex\right\vert}}

\usepackage{scalerel}
\usepackage{tikz}
\usetikzlibrary{svg.path}
\definecolor{orcidlogocol}{HTML}{A6CE39}
\tikzset{
	orcidlogo/.pic={
		\fill[orcidlogocol] svg{M256,128c0,70.7-57.3,128-128,128C57.3,256,0,198.7,0,128C0,57.3,57.3,0,128,0C198.7,0,256,57.3,256,128z};
		\fill[white] svg{M86.3,186.2H70.9V79.1h15.4v48.4V186.2z}
		svg{M108.9,79.1h41.6c39.6,0,57,28.3,57,53.6c0,27.5-21.5,53.6-56.8,53.6h-41.8V79.1z M124.3,172.4h24.5c34.9,0,42.9-26.5,42.9-39.7c0-21.5-13.7-39.7-43.7-39.7h-23.7V172.4z}
		svg{M88.7,56.8c0,5.5-4.5,10.1-10.1,10.1c-5.6,0-10.1-4.6-10.1-10.1c0-5.6,4.5-10.1,10.1-10.1C84.2,46.7,88.7,51.3,88.7,56.8z};
	}
}
\newcommand\orcidicon[1]{\href{https://orcid.org/#1}{\mbox{\scalerel*{
				\begin{tikzpicture}[yscale=-1,transform shape]
				\pic{orcidlogo};
				\end{tikzpicture}
			}{|}}}}

\begin{document}


\title{Weak and Semi-Contraction for Network Systems and Diffusively-Coupled
  Oscillators}

%
%

\author{Saber
  Jafarpour~\textsuperscript{\orcidicon{0000-0002-7614-2940}},~\IEEEmembership{Member,~IEEE,} Pedro Cisneros-Velarde,~\IEEEmembership{Student Member, IEEE,} 
   and \\ Francesco Bullo~\textsuperscript{\orcidicon{0000-0002-4785-2118}},~\IEEEmembership{Fellow,~IEEE}
   \IEEEcompsocitemizethanks{
     \IEEEcompsocthanksitem This work was supported in part by the U.S.\ Defense
     Threat Reduction Agency under grant  HDTRA1-19-1-0017.
     \IEEEcompsocthanksitem Saber Jafarpour, Francesco
     Bullo, and Pedro Cisneros-Velarde are with the Center
     of Control, Dynamical Systems and Computation, UC Santa Barbara, CA
     93106-5070, USA. {\tt \{saber,bullo\}@ucsb.edu, {pacisne@gmail.com}}.}}

%
%

\markboth{Submitted to IEEE Transactions on Automatic Control, May 18, 2020}%
{Jafarpour \MakeLowercase{\textit{et al.}}}
%


\maketitle

\begin{abstract}
  We develop two generalizations of contraction theory,
  {\color{black}namely,} semi-contraction and weak-contraction theory.
  First, using the notion of semi-norm, we propose a geometric framework
  for semi-contraction theory. We introduce matrix semi-measures and
  characterize their properties.
  {\color{black}We show that the spectral abscissa of a matrix is the
    infimum over weighted semi-measures.}
  For dynamical systems, we use the semi-measure of their Jacobian to
  {\color{black}characterize the contractivity properties of their
    trajectories.}
  Second, for weakly contracting systems, we prove a dichotomy for the
  asymptotic behavior of their trajectories and {\color{black}novel
    sufficient conditions for convergence to an equilibrium.}
  Third, we show that every trajectory of a doubly-contracting system,
  {\color{black} i.e., a system that is both weakly and semi-contracting,}
  converges to an equilibrium point.
  Finally, we apply our results to various important network systems
  including affine averaging and affine flow systems, continuous-time
  distributed primal-dual algorithms, and {\color{black}networks of
    diffusively-coupled dynamical systems}.
  {\color{black}For diffusively-coupled systems, the semi-contraction theory
    leads to a sufficient condition for synchronization that is sharper, in
    general, than previously-known tests.}
\end{abstract}


%
\IEEEpeerreviewmaketitle

\section{Introduction}

\paragraph*{Problem description} 

Strict contractivity is a useful and classical property of dynamical
systems, which ensures global exponential stability of a unique
equilibrium.  However, numerous example applications fail to satisfy this
property and exhibit richer dynamic properties. In this paper, motivated by
applications in network systems, we study systems that satisfy relaxed
versions of the standard contractivity conditions. We characterize the
implications of these relaxed conditions on the asymptotic behavior of the
dynamical system.  We aim to develop a generalized contractivity theory
that can explain the {\color{black}asymptotic} behavior of some classic
example systems, including affine averaging and flow systems, distributed
primal-dual dynamics, and {\color{black}networks of diffusively-coupled
  systems}.  {\color{black}Specifically, we aim to provide sharp conditions
  for exponential convergence and synchronization in network systems.}


\paragraph*{Literature review}

Studying contractivity of dynamical systems using matrix measures has a
long history that can be traced back to Lewis~\cite{DCL:49} and
Demidovi\v{c}~\cite{BPD:61a}. In the control community, {\color{black}
  matrix measures were adopted by Desoer and
  Vidyasagar~\cite{CAD-MV:1975,MV:78-book} and} contraction theory was
introduced by Lohmiller and Slotine~\cite{WL-JJES:98}. We refer
to~\cite{AP-AP-NVDW-HN:04} for a historical review and to the
surveys~\cite{ZA-EDS:14b,MdB-DF-GR-FS:16} for recent developments and
applications to consensus and synchronization in complex networks.


Several generalizations of contraction theory have been proposed in the
literature. In~\cite{JJS:03}, the notion of partial contraction is
introduced to study convergence of system trajectories to a specific
behavior or to a manifold. The idea is to impose contractivity only on a
part of the states of the system. Partial contraction with respect to the
$\ell_2$-norm has been further developed in~\cite{WW-JJES:05,QCP-JJS:07} to
study synchronization in complex networks. A similar notion is studied
by~\cite{MdB-DF-GR-FS:16} in the context of convergence to invariant
subspaces.  Extensions of contraction theory to non-Euclidean norms
{\color{black}and metrics} have been explored in the context of monotone
dynamical systems. For compartmental systems, contractivity with respect to
the $\ell_1$-norm has been used to study convergence to the equilibrium
point~\cite{EL-GC-KS:14}. {\color{black}For monotone systems,
  \cite{FF-RS:16} uses the contractivity of the so-called Hilbert metric to
  propose a nonlinear generalization of the Perron\textendash{}Frobenius
  theorem}. The connection between monotonicity and contractivity of
dynamical systems has been studied in detail in~\cite{SC:19,YK-BB-MC:20}.
Other extensions include contraction theory on Finsler
manifolds~\cite{FF-RS:14}, contraction theory on Riemannian
manifolds~\cite{JWSP-FB:12za}, transverse contraction for convergence to
limit cycles~\cite{IRM-JES:14}, and contraction after
transient~\cite{MM-EDS-TT:16}.


{\color{black}Contraction theory has been used to study consensus and
  synchronization problems for dynamical systems over networks. We next
  review some important example systems.}


{\color{black}Diffusively-coupled dynamics appear in different areas of
  science and engineering. Examples include (i) \emph{chemical
    reaction-diffusion} in biological tissues and the process of
  morphogenesis in developmental biology~\cite{AMT:52}, (ii) variants of
  the well-known \emph{Goodwin model} for oscillating autoregulated genes
  in cellular systems~\cite{BCG:65}, (iii) the well-known
  \emph{FitzHugh\textendash{}Nagumo model} describing neuronal interactions
  in the brain~\cite{RFH:61}, and (iv) \emph{cellular neural networks
    (CNNs)} for real-time large-scale signal processing in parallel
  computing~\cite{LOC-LY:88}.  Synchronization is arguably one of the most
  important collective behavior in networks of diffusively-coupled
  dynamics. Finding sharp conditions that ensure synchronization of
  diffusively-coupled dynamics is important for detecting stable pattern
  formations in morphogenesis, analyzing oscillatory behaviors in Goodwin
  model of cellular systems, and preventing disorders such as Parkinson’s
  disease in FitzHugh\textendash{}Nagumo model of neurons. Synchronization
  of diffusively-coupled dynamics has been extensively studied using
  contraction theory, {\color{black} e.g.,
    see~\cite{CWW-LOC:95b,WL-TC:06,WW-JJES:05,LS-RS:09,PD-MdB-GR:11,ZA-EDS:14}}.
  Indeed the notion of partial contraction was developed
  in~\cite{WW-JJES:05} precisely to study this class of problems. The early
  work~\cite{CWW-LOC:95b} introduces what is now known as the QUAD
  condition for studying synchronization of coupled oscillators.
  In~\cite{WL-TC:06}, local and global synchronization of linearly-coupled
  oscillators over directed graphs have been analyzed using the QUAD
  condition.  While explaining the relationship between QUAD condition and
  contractivity of vector fields, \cite{PD-MdB-GR:11} also studies
  diffusively-coupled identical nonlinear oscillators on undirected graphs.
  \cite{MA:11} proposes an $\ell_2$-norm condition for synchronization of
  diffusively-coupled oscillators with time-invariant interconnections. In
  many diffusively-coupled systems including chemical reaction-diffusion
  and Goodwin model, it has been shown~\cite{ZA-EDS:14,MC:07a} that
  quadratic Lyapunov functions and $\ell_2$-matrix measures provide
  conservative estimates for the onset of synchronization. In other
  diffusively-coupled systems including CNNs, certain properties of the
  internal dynamics make it easier to study synchronization with respect to
  non-Euclidean norms. For diffusively-coupled networks with time-varying
  interconnections, synchronization has been studied using non-Euclidean
  matrix measures in~\cite{MC:07a,ZA-EDS:14}. In~\cite{GR-MDB-EDS:13}, the
  synchronization of complex networks is studied using a contraction-based
  hierarchical approach via mixed norms.}

{\color{black}Primal-dual algorithms for centralized and distributed
  optimization have been widely studied and adopted in several
  applications; we refer to the~\cite{TY-XY-JW-DW:19} for a
  comprehensive survey.  In the last decade, there has been a renewed
  interest in convergence analysis~\cite{DF-FP:10} with more attention
  devoted to convergence guarantees in distributed
  optimization~\cite{JW-NE:11}. Estimates of convergence rates have
  been obtained only under strong convexity assumptions on the cost
  functions. In many important applications, obtaining convergence
  rates of these algorithms is crucial in order to provide safety
  guarantees and to analyze their discrete-time
  implementations~\cite{GQ-NL:19}. However, in some applications
  including power grids~\cite{NL-CL-ZC-SHL:13} and resource
  allocation~\cite{DD-MRJ:18}, the cost functions do not
  satisfy strong convexity assumptions and the existing literature does
  not provide estimates for their convergence rates. Recently,
  contraction theory has been used to study convergence and robustness
  of continuous-time primal-dual
  algorithms~\cite{HDN-TLV-KT-JJS:18,PCV-SJ-FB:19r} .}

\paragraph*{Contribution}

In this paper, we develop two generalizations of contraction theory that
broaden its range of applicability and allow for richer and more complex
dynamic behaviors. {\color{black}In the first generalization, called
  semi-contractivity, the notion of matrix measure is extended to
  matrix semi-measure. This allows the distance
between trajectories to increase in certain directions, thus requiring the
dynamical system to be contractive only on a certain subspace. In the
second generalization, called weak contractivity, the matrix measure
of the system Jacobian can vanish, thus allowing the distance between trajectories to
remain constant for some time.}

Using the notion of semi-norm, we present a geometric framework for studying
semi-contracting systems. We introduce the semi-measure of a matrix, which
is associated with a semi-norm, and we study the linear algebra of
semi-measures. We provide two optimization problems that establish a
connection between semi-measures of a matrix and its spectral abscissa on
certain subspaces. These optimization problems can be considered as
extensions of~\cite[Proposition 2.3]{MYL-LW:98} to semi-norms and
semi-measures, and they play a crucial role in studying the convergence
rate of trajectories in our framework. Next, we prove a generalization of
the well-known Coppel's inequality~\cite{WAC:1965}, which provides upper
and lower bounds on the semi-norm of flows of linear time-varying systems.
{\color{black} These results generalize the classic treatment of matrix
  measures and Coppel's inequality given for example
  in~\cite[Chapter~2]{MV:78-book}.}
{\color{black} Finally, for a given time-varying dynamical system and
  semi-norm, we introduce the notion of semi-contraction of the system and
  two notions of invariance of the semi-norm, namely shifted- and
  infinitesimal invariance.  We prove various results.  First, for a
  semi-contracting system with an infinitesimally invariant semi-norm, the
  semi-distance between any two trajectories vanishes exponentially fast.
  Second, for a semi-contracting system with a shifted-invariant semi-norm,
  every trajectory converges to the shifted kernel of the semi-norm.  Third
  and final, an infinitesimally invariant semi-norm is also
  shifted-invariant for a time-invariant dynamical system.}

{\color{black}The notion of semi-contraction is related to the notions of
  (i) partial contraction, as proposed in~\cite{JJS:03}, elaborated
  in~\cite{WW-JJES:05,QCP-JJS:07}, and surveyed in~\cite{MdB-DF-GR-FS:16};
  and (ii) horizontal contraction on Finsler manifolds, as proposed
  in~\cite{FF-RS:14} and elaborated in~\cite{DW:20}. Beside providing a
  comprehensive unifying framework, our treatment of semi-contraction
  differs from partial contraction since it allows arbitrary semi-measures
  and from horizontal contraction as it leads to sharper statements for a
  more restricted class of systems.  Appendix~\ref{app:semi-horizontal}
  contains a detailed comparison between semi-contraction and horizontal
  contraction.}


{\color{black} Weakly contracting (also referred to as non-expanding)
  dynamical systems are introduced, and their asymptotic behavior analyzed,
  in~\cite{GC:17,SC:19,PCV-SJ-FB:19r}.  These treatments, however, are
  based on specific norms and tailored to specific applications. In this
  paper, we develop two novel general results about the asymptotic behavior
  of weakly contracting systems.}  Our first result is a dichotomy for
{\color{black}systems that are weakly contracting with respect to an
  arbitrary norm}; each trajectory is bounded if and only if the system has
an equilibrium point. For systems that are weakly contracting with respect
to a weighted $\ell_1$-norm or a weighted $\ell_\infty$-norm, we show that
the dichotomy can be more elaborate; every trajectory converges to an
equilibrium point if and only if an equilibrium point exists.


We introduce the class of doubly-contracting
systems, i.e, systems which are both semi-contracting and
weakly contracting (possibly with respect to two different semi-norms). We
show that, for a doubly-contracting system whose equilibrium points form a
subspace of $\real^n$, every trajectory converges exponentially to an
equilibrium point. Moreover, we provide a convergence rate that is
independent of the norm and the semi-norm and depends only on the limiting
equilibrium point.

We provide several applications of our results to network systems: affine
averaging and flow systems, continuous-time primal-dual dynamics, and
{\color{black}networks of coupled dynamical systems. First, we study the
  affine averaging and affine flow systems using our framework. Affine
  averaging and affine flow models are the most basic examples of linear
  network systems and are embedded in numerous nonlinear network
  models. Any general methodology for the analysis of nonlinear network
  system needs to fully recover the properties of these systems. }We show
that affine averaging (respectively, affine flow) systems are weakly
contracting with respect to the $\ell_\infty$-norm (respectively,
$\ell_1$-norm) and semi-contracting with respect to an appropriate weighted
$\ell_\infty$-norm (respectively, $\ell_1$-norm).  {\color{black}Our theorem
  on doubly contracting systems then characterizes the complete asymptotic
  behavior of both systems, that is, global exponential convergence to a
  unique equilibrium or unbounded evolution as a function of the network
  parameters.}


{\color{black} Second, we study a distributed implementation of the
  continuous-time primal-dual algorithm for optimizing a function over a
  connected network. We show that, under some weak convexity assumption on
  the cost function, the distributed primal-dual algorithm is weakly
  contracting and every trajectory of the system converges exponentially to
  an equilibrium point. Additionally, we obtain an expression for the
  convergence rate to the equilibrium points.  Compared
  to~\cite{PCV-SJ-FB:19r}, we prove the global exponential convergence of
  the distributed primal-dual algorithm to the solution of the optimization
  problem and obtain the convergence rate when the cost function is weakly
  convex. }

  

Finally, we investigate synchronization in networks of diffusively-coupled
identical dynamical systems. For networks of diffusively-coupled identical
dynamics, we use our semi-contraction framework to propose novel sufficient
conditions for global synchronization. A key step in obtaining these
synchronization conditions is the introduction of a new class of norms
called $(2,p)$-tensor norms, with various useful properties. Compared to
the contraction-based approaches using the
$\ell_2$-norm~\cite{CWW-LOC:95b,WW-JJES:05,PD-MdB-GR:11}, our
synchronization conditions (i) are {\color{black}more general since they
  are} based on the weighted $\ell_p$-matrix measure of the Jacobian of
{\color{black}the dynamical system} for every $p\in [1,\infty]$, (ii)
provide an explicit rate of convergence to the synchronized trajectories,
and (iii) can recover the exact threshold of synchronization for
{\color{black}diffusively-coupled linear systems.}  Compared to the results
in~\cite{MC:07a,ZA-EDS:14}, our synchronization conditions (i) are
applicable to arbitrary undirected network topology, (ii) allow for a
arbitrary class of weighted $p$-norms, for every $p\in [1,\infty]$, and
(iii) demarcate the roles of internal dynamics and the network
connectivity.

\paragraph*{Paper Organization} Section~\ref{sec:notation} introduces
the notation. Section~\ref{sec:partial} introduces the matrix semi-measures
and use them to study the semi-contracting systems. Sections~\ref{sec:weak}
and~\ref{sec:double} study weakly contracting and doubly-contracting
systems, respectively. Finally, Section~\ref{sec:app} analyzes three
applications of our semi-and weak-contraction theory to network
systems.
This document is an ArXiv technical report and, compared with its journal
version, it additionally contains in Appendix~\ref{app:Theorem5-6-proofs} the proof of
Theorems~\ref{thm:measureporp} and~\ref{lem:computational}, in
Appendix~\ref{app:diffusive-extended} an extension of Theorem~\ref{thm:diffusive} for
diffusively-coupled identical dynamical systems, and in Appendix~\ref{sec:LV} a
contractivity-based analysis of the cooperative Lotka\textendash{}Volterra
population model.

\section{Notation}\label{sec:notation}
For a set $S\subseteq \real^n$, its interior, closure, and diameter are
denoted by $\mathrm{int}(S)$, $\mathrm{cl}(S)$, and $\mathrm{diam}(S)$,
respectively. The $n\times n$ identity matrix is $I_n$ and the all-ones and
all-zeros column vectors of length $n$ are $\vect{1}_n$ and $\vect{0}_n$,
respectively. For $A\in \complex^{n\times m}$, the conjugate transpose of
$A$ is $A^{H}$, the real part of $A$ is $\Re(A)$, the range of $A$ is
$\Img(A)$ and the kernel of $A$ is $\Ker(A)$. The Moore\textendash{}Penrose
inverse of $A$ is the unique matrix $A^{\dagger}\in \complex^{m\times n}$
such that $AA^{\dagger}A = A$, $A^{\dagger}AA^{\dagger}= A^{\dagger}$ and
$AA^{\dagger}$ and $A^{\dagger}A$ are Hermitian matrices. It can be shown
that $AA^{\dagger }$ is the orthogonal projection onto $\Img(A)$ and
$A^{\dagger}A$ is the orthogonal projection onto $\Img(A^{H})$.  Let
$\lambda_1(A),\ldots,\lambda_n(A)$ and $\mathrm{spec}(A)$ denote the
eigenvalues and the spectrum of $A \in\complex^{n\times
  n}$. {\color{black}Given two real symmetric matrices $A,B\in
  \real^{n\times n}$, we write $A\preceq B$ if $B-A$ is positive semi-definite. Given a vector
subspace $\mcS\subseteq \complex^{n}$, a vector $v\in \complex^n$, and
a matrix $A\in \complex^{n\times n}$, the orthogonal complement of $S$ is
$S^{\perp}$ and we define $v+\mcS = \setdef{v+u}{u\in \mcS}$ and $A\mcS := \setdef{Au}{u\in \mcS}$.} The vector
subspace $\mcS\subseteq \complex^{n}$ is invariant under $A\in
\complex^{n\times n}$ if $A\mcS\subseteq\mcS$. Given $A\in
\complex^{n\times n}$ and a vector subspace $\mcS\subseteq \complex^{n}$,
$\mcS$ is invariant under $A$ if and only if $\mcS^{\perp}$ is invariant
under $A^{H}$. Given $A\in \complex^{n\times n}$ and a vector subspace
$\mcS\subseteq \complex^n$ invariant under $A$, define
\begin{align*}
  \mathrm{spec}_{\mcS}(A) :=
  \setdef{\lambda\in\mathrm{spec}(A)}{\exists v\in\mcS\mbox{
      s.t. } Av = \lambda v}.
\end{align*}
Note
$\mathrm{spec}(A) = \mathrm{spec}_{\mcS}(A) \cup
\mathrm{spec}_{\mcS^{\perp}}(A^{H}) $. {\color{black}Denote the
  spectral abscissa of $A$ by $\alpha(A)$} and define the spectral
abscissa of $A$ restricted to $\mcS$ by
$\alpha_{\mcS}(A)
:=\max\setdef{\Re(\lambda)}{\lambda\in\mathrm{spec}_{\mcS}(A)}$. For
$A\in \complex^{n\times n}$ such that $\Re(\lambda)\le 0$, for every
$\lambda\in \mathrm{spec}(A)$, the essential spectral abscissa of $A$
is defined by
\begin{align*}
  \alpha_{\mathrm{ess}}(A) :=\max\setdef{\Re(\lambda)}{\lambda\in\mathrm{spec}(A)-\{0\}}.
\end{align*}
The absolute value of $z\in \complex$ is denoted by $|z|$. A norm
$\|\cdot\|$ on $\complex^n$ is \emph{absolute} if $\|x\|=\|y\|$ for
every $x,y\in \complex^n$ such that $|x_i|=|y_i|$, $i\in\{1,\ldots,n\}$. Let $x\in \real^n$, $r>0$, and
$\|\cdot\|$ be a norm on $\real^n$. Then the open ball of $\|\cdot\|$
centered at $x$ with radius $r$ is $B_{\|\cdot\|}(x,r)=\setdef{y\in
  \real^n}{\|y-x\|<r}$ and the closed ball of $\|\cdot\|$ centered at $x$
with radius $r$ is $\overline{B}_{\|\cdot\|}(x,r)=\setdef{y\in
  \real^n}{\|y-x\|\le r}$. 
For $p\in [1,\infty]$, denote the $\ell_p$-norm on $\real^n$ by
$\|\cdot\|_p$.  All $\ell_p$-norm are absolute. The
$\ell_p$-norm is polyhedral if and only if $p\in \{1,\infty\}$.  For any
two complex matrices $A$ and $B$, their Kronecker product is denoted by
$A\otimes B$. Given two vector spaces $V$ and
$W$, the tensor product space $V\otimes W$ is given by $
V\otimes W = \mathrm{span}\setdef{v\otimes w}{v\in V,\; w\in W}$. Consider the time-varying dynamical system:
\begin{align}\label{eq:TV-nonlinear}
  \dot{x} = f(t,x),\quad t\in \real_{\ge 0}, x\in \real^n.
\end{align}
We assume that $(t,x)\mapsto f(t,x)$ is twice-differentiable in $x$ and essentially
bounded in $t$. Denote the flow of~\eqref{eq:TV-nonlinear} starting from $x_0$ by $t\mapsto
\phi(t,x_0)$. A set $\mcS\subseteq \real^n$ is invariant with respect
to~\eqref{eq:TV-nonlinear} if $\phi(t,\mcS)\subseteq \mcS$ for every $t\in
\real_{\ge 0}$.  A vector field $\map{X}{\real^n}{\real^n}$ is piecewise
real analytic if there exist closed sets $\{\Sigma_i\}_{i=1}^{m}$ which
partition $\real^n$ and $X$ is real analytic on $\mathrm{int}(\Sigma_i)$,
for every $i\in \{1,\ldots,m\}$.

\section{Semi-contracting systems}\label{sec:partial}

{\color{black}In this section we provide a
geometric framework for semi-contracting systems via semi-norms}.

\subsection{Linear algebra and matrix semi-measures}\label{sec:semi-measure}

We start with semi-norms and associated semi-measures.

\begin{definition}[Semi-norms]
  A function $\verti{\cdot}:\real^n \to \real_{\ge 0}$ is a
  \emph{semi-norm} on $\real^n$, if 
  \begin{enumerate}
  \item $\verti{cv} = |c| \verti{v}$, for every $v\in \real^n$ and
    $c\in \real$;
  \item $\verti{v+w} \le \verti{v}+\verti{w}$, for every $v,w\in
    \real^n$. 
  \end{enumerate}
\end{definition}
For a semi-norm $\verti{\cdot}$, its kernel is defined by
\begin{align*}
  \Ker \verti{\cdot} = \setdef{v\in \real^n}{\verti{v} = 0}. 
\end{align*}
It is easy to see that $\Ker \verti{\cdot}$ is a subspace of $\real^n$ and
$\verti{\cdot}$ is a norm on the vector space
$\Ker\verti{\cdot}^{\perp}$. Semi-norms can naturally arise from norms. Let
$k\le n$, and $\|\cdot\|:\real^k\to \real_{\ge 0}$ be a norm on $\real^k$
and $R\in \real^{k\times n}$. Then the $R$-weighted semi-norm on $\real^n$
associated with the norm $\|\cdot\|$ on $\real^k$ is defined by {\color{black}
\begin{align}\label{eq:R-seminorm}
  \verti{v}_{R} = \|Rv\|,\quad\mbox{ for all } v\in \real^n .
\end{align}}
It is easy to see that the $R$-weighted semi-norm {\color{black}$\verti{\cdot}_{R}$} is a norm
if and only if $k=n$ and $R$ is invertible.

A semi-norm $\verti{\cdot}$ on $\real^n$
naturally induces a semi-norm on the space of real-valued matrices
$\real^{n\times n}$.

\begin{definition}[Induced semi-norm]\label{def:induced-norm}
  Let $\verti{\cdot}:\real^n \to \real_{\ge 0}$ be a semi-norm on
  $\real^n$, the induced semi-norm on $\real^{n\times n}$ (which without
  any confusion, we denote again by $\verti{\cdot}$) is defined by
  \begin{align*}
    \verti{A} = \sup\setdef{\verti{Av}}{\verti{v}=1, \;
      v\perp\Ker\verti{\cdot}}.
  \end{align*}
\end{definition}

The following properties of induced semi-norms are known~\cite{VVK:83} and
we omit the proof in the interest of brevity.
      
\begin{proposition}[Properties of induced semi-norms]\label{induced-semi-norm}
      Let $\verti{\cdot}:\real^n \to \real_{\ge 0}$ be a semi-norm on
      $\real^n$ and denote the induced semi-norm on $\real^{n\times n}$
      again by $\verti{\cdot}$. Then, for every $A,B\in \real^{n\times n}$
      and $c\in \real$,
    \begin{enumerate}
    \item\label{p1:trivial} $\verti{I_n} = 1$, $\verti{A}\ge 0$, and
      $\verti{cA} = |c| \verti{A}$;
    \item\label{p3:triangle1} $\verti{A+B}\le \verti{A} +\verti{B}$;
    \item\label{p5:triangle2} $\verti{Av}\le \verti{A}\verti{v}$, for every
      $v\perp \Ker\verti{\cdot}$.
      \end{enumerate}
    \end{proposition}
    
      \begin{definition}[Matrix semi-measures]
         Let $\verti{\cdot}:\real^n \to \real_{\ge 0}$ be a semi-norm on
         $\real^n$ and we denote the induced semi-norm on $\real^{n\times
           n}$ again by $\verti{\cdot}$. Then the \emph{matrix
           semi-measure} associated with $\verti{\cdot}$ is defined by
    \begin{align}\label{eq:matrix-seminorm}
      \mu_{\verti{\cdot}} (A) = \lim_{h\to 0^+} \frac{\verti{I_n + hA} -
        1}{h}.
      \end{align}
 \end{definition}

 \begin{theorem}[Properties of matrix semi-measures]\label{thm:measureporp}
   Let $\verti{\cdot}:\real^n \to \real_{\ge 0}$ be a semi-norm on
   $\real^n$ and let $\mu_{\verti{\cdot}}$ be the associated matrix
   semi-measure. Then, for every $A,B\in \real^{n\times n}$, 
   \begin{enumerate}
   \item\label{p1:well-def} $\mu_{\verti{\cdot}}(A)$ is well-defined;
   \item \label{p2:tri} $\mu_{\verti{\cdot}}(A+B) \le
     \mu_{\verti{\cdot}}(A) + \mu_{\verti{\cdot}}(B)$;
   \item\label{p4:lipschitz}
     $|\mu_{\verti{\cdot}}(A)-\mu_{\verti{\cdot}}(B)|\le \verti{A-B}$.
   \end{enumerate}
   Moreover, if $\Ker\verti{\cdot}$ is invariant under $A$, then
  \begin{enumerate}\setcounter{enumi}{3}
  \item\label{p3:realpart} $\alpha_{\Ker\verti{\cdot}^{\perp}}(A^\top)\le
    \mu_{\verti{\cdot}}(A)$.
   \end{enumerate}
 \end{theorem}
 \begin{proof}
   We refer to Appendix~\ref{app:Theorem5-6-proofs} for the proof. 
   \end{proof}

Given the $R$-weighted semi-norm {\color{black}$\verti{\cdot}_{R}$} as defined
in~\eqref{eq:R-seminorm}, we denote its induced semi-norm on
$\real^{n\times n}$ by {\color{black}$\verti{\cdot}_{R}$} and its matrix semi-measure by
$\mu_{R}$. Specifically, for the $\ell_p$-norm on $\real^k$, we let
{\color{black}$\verti{\cdot}_{p,R}$} and $\mu_{p,R}$ denote the associated $R$-weighted
semi-norm on $\real^n$ and matrix semi-measure, respectively.
 
\begin{theorem}[Computation of semi-measures]\label{lem:computational}
    Let $\|\cdot\|$ be a norm with associated matrix measure $\mu$, $R\in
  \real^{k\times n}$, for $k\le n$, be a full rank matrix,
  {\color{black} and $P=R^{\top}R\in \real^{n\times n}$.} Then, for each
  $A\in \real^{n\times n}$, $\xi\in \real^n_{\ge 0}$, {\color{black}and $c\in
    \real$},
  \begin{enumerate}
  \item\label{p1:seminorm} {\color{black}$\verti{A}_R = \|RAR^{\dagger}\|$},
  \item\label{p2:seminorm} $\mu_{R} (A) = \mu(RAR^{\dagger})$,
  \item\label{p4:1-norm}$\mu_{1,\diag(\xi)}(A) =\max_{j:\xi_j\ne
      0}\left\{a_{jj} + \xi_j\sum_{i: \xi_i\ne 0}
      \frac{|a_{ij}|}{\xi_i}\right\}$,
  \item\label{p4:inf-norm}$\mu_{\infty,\diag(\xi)}(A) =\max_{i:\xi_i\ne
      0}\left\{a_{ii} + \xi_i\sum_{j: \xi_j\ne 0}
      \frac{|a_{ij}|}{\xi_j}\right\}$,
  \item\label{p5:2-norm} {\color{black}$\mu_{2,R}(A)\le c$ if and only
      if, for each $x\in \Ker(P)^{\perp}$,}
      \begin{align}\label{eq:semi-demidovich}
        x^{\top}(PA+A^{\top}P - 2c P)x\le 0,
      \end{align}
\end{enumerate}
Moreover, if $\Ker(R)$ is invariant under $A$, then
\begin{enumerate}\setcounter{enumi}{5}
  \item\label{p6:2-norm-inv} {\color{black}$\mu_{2,R}(A)\le c$ if and only
      if $PA+A^{\top}P \preceq 2c P$,}
\item \label{p3:2-norm} $\mu_{2,R}(A) =
  \frac{1}{2}\alpha_{\Ker R^{\perp}}\left(A+ P^{\dagger}A^{\top}P\right)$.
  \end{enumerate}
\end{theorem}
\begin{proof}
   We refer to Appendix~\ref{app:Theorem5-6-proofs} for the proof. 
   \end{proof}

  \subsection{Spectral abscissa as an optimal matrix measure}\label{sec:abscissa}

  Theorem~\ref{thm:measureporp}\ref{p3:realpart} shows that the
  semi-measures of a matrix is lower bounded by its spectral abscissa. In
  the next theorem we study this gap and we show that on the space of all
  semi-measures this lower bound is tight. In this part, we use the
  generalization of the results in Section~\ref{sec:semi-measure} to
  $\complex^n$. Specifically, we use
  Theorem~\ref{lem:computational}\ref{p1:seminorm} and~\ref{p2:seminorm} and
  and Theorem~\ref{thm:measureporp}\ref{p3:realpart} for matrices,
  norms, and matrix measures defined on $\complex^n$. 

\begin{theorem}[Optimal matrix measures and spectral abscissa]\label{thm:abscissa}
  Let $A\in\complex^{n\times{n}}$ and let $\mcS\subseteq\complex^n$ be a
  $(n-k)$-dimensional subspace which is invariant under $A$. Then
  \begin{enumerate}[leftmargin=4mm]
  \item\label{omm:p1:semi-norm} $\alpha_{\mcS^{\perp}}(A^H) =
    \inf\setdef{\mu_{\verti{\cdot}}(A)}{\verti{\cdot} \mbox{ a semi-norm
        with kernel $\mcS$}}$;
  \item\label{omm:p2:R} let $\|\cdot\|$ be an absolute norm with its
    associated matrix measure $\mu$, then $\alpha_{\mcS^{\perp}}(A^{H}) =
    \inf \setdef{\mu_R(A)}{R\in \complex^{k\times n},\Ker(R) = \mcS}$.
    \end{enumerate}
\end{theorem}
\begin{proof}
  First note that, by Theorem~\ref{thm:measureporp}\ref{p3:realpart}, for
  every semi-norm $\verti{\cdot}$ with $\Ker\verti{\cdot}=\mcS$, we have
  $\alpha_{\mcS^{\perp}}(A^H)\le \mu_{\verti{\cdot}}(A) $. This implies that
  \begin{align}\label{eq:inequality1}
     \alpha_{\mcS^{\perp}}(A^H) & \le \inf\setdef{\mu_{\verti{\cdot}}(A)}{\verti{\cdot}
      \mbox{ a semi-norm }, \Ker\verti{\cdot}=\mcS}\nonumber\\
    &\le \inf \setdef{\mu_R(A)}{R\in
     \complex^{k\times n}, \Ker(R) = \mcS} .
    \end{align}
  Therefore, in order to prove statements~\ref{omm:p1:semi-norm}
  and~\ref{omm:p2:R}, we need to show that
  \begin{align*}
   \alpha_{\mcS^{\perp}}(A^H) &\ge \inf \setdef{\mu_R(A)}{R\in
     \complex^{k\times n},\Ker(R) = \mcS}\\ & \ge \inf\setdef{\mu_{\verti{\cdot}}(A)}{\verti{\cdot}
    \mbox{ a semi-norm }, \Ker\verti{\cdot}=\mcS}.
    \end{align*}
  We start by proving statement~\ref{omm:p2:R}. Consider the case that $A$ is
  diagonalizable. Let  $\mathrm{spec}_{\mcS^{\perp}}(A^{H}) =
  \{\lambda_1,\ldots,\lambda_{k}\}$. Since $A$ is diagonalizable, there
  exists a set of linearly independent vectors
  $\{v_1,\ldots,v_{k}\}\subset \mcS^{\perp}$ such that
  $A^Hv_i=\lambda_i v_i$. Define the matrix $R\in
  \complex^{k\times n}$, where $R_i$ (the $i$th row of matrix $R$) is
  equal to $v_i^{H}$, for every $i\in \{1,\ldots, k\}$. Note that
  $\{v_1,\ldots,v_{k}\}$ is linearly independent and therefore $R$ is
  full rank. Moreover, it is easy to see that $\Ker R =\mcS$. On the other
  hand, we have $RA = \Lambda R$, where
  $\Lambda=\diag\{\lambda^H_1,\ldots,\lambda^H_{k}\}\in
  \complex^{k\times k}$. This implies that
  $RAR^{\dagger}=\Lambda$. As a result, we have $\mu_{R}(A) =
  \mu(RAR^{\dagger}) = \mu(\Lambda)$ and therefore
  \begin{align*}
    \mu_{R}(A) &= \mu(\Lambda) = \lim_{h\to 0^+} \frac{\|I_{k} +
    h\Lambda\|-1}{h} \\ & =
    \max_{i\in\until{k}}\{\Re(\lambda_i)\} = \alpha_{\mcS^{\perp}}(A^H),
  \end{align*}
  where the third equality holds because $\|\cdot\|$ is absolute and so
  $\|I_n+h\Lambda\| = \max_{i\in\until{k}}\{|1+ h\lambda_i|\}$.
  Thus
  $\alpha_{\mcS^{\perp}}(A^H) \ge \inf \setdef{\mu_R(A)}{R\in
    \complex^{k\times n},\Ker(R) = \mcS}$. Now, consider the case when $A$
  is not diagonalizable. Note that, in the complex field $\complex$, the
  set of diagonalizable matrices are dense in $\complex^{n\times
    n}$. Therefore, for every $\epsilon>0$, there exists a diagonalizable
  $A_{\epsilon}$ such that $\|A-A_{\epsilon}\|\le \epsilon$.
  Theorem~\ref{thm:measureporp}\ref{p4:lipschitz} implies
  $|\mu_{R}(A)-\mu_{R}(A_{\epsilon})|\le \epsilon\|R\|\|R^{\dagger}\|$.  As
  a result
  \begin{align*}
    \mu_{R}(A)\le \mu_{R}(A_{\epsilon}) + \epsilon\|R\|\|R^{\dagger}\|.
  \end{align*}
  By taking the infimum over the set of $R\in
  \complex^{k\times n}$ with $\Ker(R)=\mcS$
  and noting the fact that $\sup\{\|R\|\|R^{\dagger}\|\mid R\in
  \complex^{k\times n}, \;
  \Ker(R)=\mcS\}\le M$, for some $M\in \real_{\ge 0}$, we get
  \begin{align*}
     \inf&\{\mu_R(A)\mid R\in \complex^{k\times n},\Ker(R) = \mcS\} \\
    & \qquad \le  \inf\{\mu_R(A_{\epsilon})\mid R\in
     \complex^{k\times n},\Ker(R) =
     \mcS\} + M\epsilon \\
     & \qquad = \alpha_{\mcS^{\perp}}(A^H_{\epsilon}) + M\epsilon,
  \end{align*}
  where the last equality holds because $A_{\epsilon}$ is
  diagonalizable. By continuity of eigenvalues, we get
  $\lim_{\epsilon\to0^+} \alpha_{\mcS^{\perp}}(A^H_{\epsilon}) =
  \alpha_{\mcS^{\perp}}(A^H)$. This implies that
  \begin{align*}
    \inf\{\mu_R(A)\mid R\in \complex^{k\times n},\Ker(R) = \mcS\} \le \alpha_{\mcS^{\perp}}(A^H).
  \end{align*}
  This completes the proof of statement~\ref{omm:p2:R}.
  Statement~\ref{omm:p1:semi-norm} is a consequence
  of~\eqref{eq:inequality1} and statement~\ref{omm:p2:R}.
\end{proof}
   
\begin{remark}
  \begin{enumerate}
  \item Statement~\ref{omm:p1:semi-norm} is a generalization to semi-norms
    and semi-measures of \cite[Proposition 2.3]{MYL-LW:98}, which states
    $\alpha(A) = \inf\setdef{\mu_{\norm{\cdot}}(A)}{\norm{\cdot} \mbox{ a
        norm}}$.
      
  \item Statement~\ref{omm:p2:R} is a generalization to absolute
    norms of~\cite{MV:78}, which essentially proves the result for the
    $\ell_2$ case.  For the special case of $\ell_p$-norm on $\real^n$,
    Theorem~\ref{thm:abscissa} shows that the infimum of the weighted
    $\ell_p$-measure of a matrix with respect to the weights recovers the
    spectral abscissa of the matrix, that is, for any $p\in [1,\infty]$,
      \begin{equation}
        \alpha(A) = \inf_{R \text{ invertible}} \mu_{p,R}(A). \eqoprocend
      \end{equation}
    \end{enumerate} 
  \end{remark}

Next we present an application to algebraic graph theory.

\begin{lemma}[Semi-measures of Laplacian matrices]\label{ex:useful}
  Let $G$ be a weighted digraph with a globally reachable vertex and with
  Laplacian matrix $L$ (with eigenvalue $\lambda_1(L)=0$). Let $\|\cdot\|$ be
  a norm with associated matrix measure $\mu$. Then
  \begin{multline} \label{prop:Reps}
    \inf\setdef{\mu_{R}(-L)}{R\in \complex^{(n-1)\times n},
      \Ker(R)=\mathrm{span}(\vect{1}_n)} \\
    =\alpha_{\vect{1}_n^{\perp}}(-L^{\top})= \alpha_{\mathrm{ess}}(-L)<0.
  \end{multline}
  Moreover, if $G$ is undirected, then
  \begin{multline} \label{prop:RV}
    \min\setdef{\mu_{R}(-L)}{R\in\complex^{(n-1)\times n},
      \Ker(R)=\mathrm{span}(\vect{1}_n)} \\
    =   \mu_{R_\mcV}(-L) = \alpha_{\vect{1}_n^{\perp}}(-L^{\top})= -\lambda_2(L) <0,
  \end{multline}
  where $\mcV=\{v_2,\dots,v_n\}$ are orthonormal eigenvectors of the Laplacian $L$ and
  \begin{equation} \label{def:RV}   
    R_{\mcV} = \begin{bmatrix} v_2 & \dots & v_n \end{bmatrix}^\top \in \real^{(n-1)\times n}.
  \end{equation}
\end{lemma}
\begin{proof}
  Theorem~\ref{thm:abscissa}\ref{omm:p2:R} with $A=-L$ and $\mcS = \mathrm{span}(\vect{1}_n)$ gives
  \begin{align*}
    \inf_{\substack{R\in \complex^{(n-1)\times n}\\\Ker(R)=\mathrm{span}(\vect{1}_n)}} \mu_{R}(-L) =\alpha_{\vect{1}_n^{\perp}}(-L)
  \end{align*}
  Since $G$ has a globally reachable node, $\lambda_1 = 0$ is a simple
  eigenvalue of $L$ with associated right eigenvector $\vect{1}_n$ and all
  the other eigenvalues have positive real parts. This means that
  $\alpha_{\vect{1}_n^{\perp}}(-L^{\top})=\alpha_{\mathrm{ess}}(-L)<0$. 

  If $G$ is undirected, then $L$ is symmetric and all its eigenvalues are
  real. Let $0=\lambda_1<\lambda_2\le \lambda_3\le \ldots \le \lambda_n$
  denote the eigenvalues of $L$. Then it is clear that
  $\alpha_{\vect{1}_n^{\perp}}(-L^{\top}) = \alpha_{\mathrm{ess}}(-L) =
  -\lambda_2<0$. Moreover, we have
  $R_{\mcV}LR_{\mcV}^{\top}=\diag(\lambda_2,\ldots,\lambda_n)$. Additionally,
  we know that
  $R_{\mcV}^{\dagger}=R_{\mcV}^{\top}(R_{\mcV}R_{\mcV}^{\top})^{-1} =
  R_{\mcV}^{\top}$~\cite[E4.5.20]{CDM:01}. As a result,
  \begin{align*}
    \mu_{R_{\mcV}}(-L) = \mu(-R_{\mcV}LR_{\mcV}^{\dagger}) = \mu(-R_{\mcV}LR_{\mcV}^{\top}) = -\lambda_2.
    \hfill\quad\mbox{}\qquad\qedhere
  \end{align*}
\end{proof}

\begin{remark}\label{remark:Reps}
Equation~\eqref{prop:Reps} implies that, for any $\epsilon>0$, there exists
a matrix $R_\epsilon\in \complex^{(n-1)\times n}$ with
$\Ker(R)=\mathrm{span}(\vect{1}_n)$ satisfying
$\mu_{R_\epsilon}(-L)\leq\alpha_{\mathrm{ess}}(-L)+\epsilon$. \oprocend
\end{remark}
 
\subsection{{\color{black}Semi-contraction results for dynamical systems}}
  Consider a continuous map $t\mapsto A(t)\in \real^{n\times n}$ and the
  dynamical system
  \begin{align}\label{eq:LTV}
    \dot{x}(t) = A(t) x(t),\quad\mbox{ for } t\ge t_0\in \real
  \end{align}
  with the initial condition $x(t_0)=x_0$. 

  \begin{theorem}[{\color{black}Coppel's inequality for semi-norms}]\label{thm:coppel}
    Let $\verti{\cdot}$ be a semi-norm on $\real^n$ and
    let $\mu_{\verti{\cdot}}$ be the associated matrix
    semi-measure. Assume that, for every $t\ge t_0$,
    $\Ker\verti{\cdot}$ is invariant under $A(t)$. Then, for every $t\ge t_0$, we have
    \begin{multline*}
      \exp\left(\int_{t_0}^{t} \mu_{\verti{\cdot}}(-A(\tau))d\tau\right)\verti{x(0)} \le
      \verti{x(t)}\\\le \exp\left(\int_{t_0}^{t} \mu_{\verti{\cdot}}(A(\tau))d\tau\right) \verti{x(0)}.
      \end{multline*}
      Moreover, if $A$ is time-invariant, we have
      \begin{align*}
        \exp(t\mu_{\verti{\cdot}}(-A)) \verti{x(0)} \le
      \verti{x(t)}\le \exp(t\mu_{\verti{\cdot}}(A)) \verti{x(0)}. 
      \end{align*}
\end{theorem}
\begin{proof}
  Note that $t\mapsto A(t)$ is continuous. Therefore the solutions
  $t\mapsto x(t)$ of the time-varying dynamical system~\eqref{eq:LTV}
  are differentiable. Thus, for small enough $h$, we can write
  $x(t+h) = x(t) + h A(t) x(t) + o(h)$. Let  $\mathcal{P}$ be
  the orthogonal projection onto $\Ker\verti{\cdot}^{\perp}$. Then we
  have
  \begin{align*}
    \mathcal{P} x(t+h) &= \mathcal{P} x(t) + h \mathcal{P} A(t) x(t) +
    o(h) \\ & = \mathcal{P} x(t) + h \mathcal{P} A(t) \mathcal{P} x(t) +
    o(h).
  \end{align*}
  where the last equality holds because $A(t) \Ker\verti{\cdot} \subseteq
  \Ker\verti{\cdot}$. Therefore, $\verti{\mathcal{P}v} = \verti{v}$ and
  Theorem~\ref{induced-semi-norm} together imply
  \begin{align*}
    \frac{\verti{x(t+h)} - \verti{x(t)}}{h} \le  \frac{\verti{I_n +
    hA} -1}{h} \verti{x(t)} + \frac{o(h)}{h}.
  \end{align*}
  Taking the limit as $h\to 0^+$, we obtain $\frac{d}{dt} \verti{x(t)} \le
  \mu_{\verti{\cdot}}(A)\verti{x(t)}$. The result then follows from the
  Gr\"{o}nwall\textendash{}Bellman inequality.
\end{proof}

{\color{black}\begin{definition}[Semi-contracting systems]
  Let $C\subseteq \real^n$ be a convex set, $c>0$ be a positive number,
  $\verti{\cdot}$ be a semi-norm on $\real^n$, and $\mu_{\verti{\cdot}}$ be
  its associated matrix semi-measure. The time-varying dynamical
  system~\eqref{eq:TV-nonlinear} is
  \begin{enumerate}
  \item   \emph{semi-contracting} on $C$ with
  respect to $\verti{\cdot}$ with rate $c$, if
  \begin{align*}
    \mu_{\verti{\cdot}}(\jac{f}(t,x)) \le -c,\enspace\mbox{for all }t\in
    \real_{\ge 0}, x\in C
  \end{align*}
  \end{enumerate}
  Moreover, the subspace $\Ker\verti{\cdot}$ is
  \begin{enumerate}
  \item \emph{infinitesimally invariant} under the
    system~\eqref{eq:TV-nonlinear}, if 
    \begin{align}\label{InfI}
      \jac f(t,x) \Ker\verti{\cdot} \subseteq \Ker\verti{\cdot},
      \enspace\mbox{for all }t\in
      \real_{\ge 0}, x\in C; 
    \end{align}
      \item \emph{shifted-invariant} under the 
    system~\eqref{eq:TV-nonlinear}, if there exists
    $x^*\in \real^n$, such that
    \begin{align} \label{ShfI}
      f(t,x^*+\Ker\verti{\cdot})\subseteq \Ker\verti{\cdot},
      \enspace\mbox{for all }t\in \real_{\ge 0}. 
      \end{align}
  \end{enumerate}
  \end{definition}

  Now, we study the asymptotic behavior of trajectories of the
  dynamical system~\eqref{eq:TV-nonlinear} when it is 
  semi-contracting.
  
  \begin{theorem}[Trajectories of semi-contracting 
  systems]\label{thm:partialcontraction} Consider the time-varying
    dynamical system~\eqref{eq:TV-nonlinear} with a convex and invariant
    set $C$. Assume the system~\eqref{eq:TV-nonlinear} is semi-contracting
    on $C$ with respect to a semi-norm $\verti{\cdot}$ with rate
    $c>0$. Then
    \begin{enumerate}      
    \item\label{p1:estimates} if infinitesimal invariance~\eqref{InfI} holds,
      then for every $x_0,y_0\in C$ and $t\in \real_{\ge 0}$,
      \begin{align}\label{eq:semi-contraction-traj}
        \verti{\phi(t,x_0) - \phi(t,y_0)} \le e^{-ct}\verti{x_0-y_0};
      \end{align}

    \item\label{p2:convergence} if shifted-invariance~\eqref{ShfI} holds for
      some $x^*\in \Ker\verti{\cdot}^{\perp}$ and $C=\real^n$, then for every $x_0\in \real^n$ and $t\in
      \real_{\ge 0}$,
      \begin{align}\label{eq:partial-contraction-traj}
        \verti{\phi(t,x_0)-x^*} \le  e^{-ct}\verti{x_0-x^*},
      \end{align}
      and $t\mapsto \phi(t,x_0)$ converges to
      $x^*+\Ker\verti{\cdot}$ with exponential rate $c$.
    \end{enumerate}
     Moreover, assuming the system~\eqref{eq:TV-nonlinear} is
     time-invariant,
      \begin{enumerate}\setcounter{enumi}{2}
    \item\label{p3:f} if \eqref{InfI} holds, then for every $x_0\in \real^n$
      and every $t\in \real_{\ge 0}$,
      \begin{align*}
        \verti{f(\phi(t,x_0))} \le e^{-ct}\verti{f(x_0)};
      \end{align*}
      \item\label{p4:invariantsubspace} if \eqref{InfI} holds and
        $C=\real^n$, then~\eqref{ShfI} holds for some $x^*\in \real^n$.
    \end{enumerate}
\end{theorem}
  \begin{proof}
      Regarding part~\ref{p1:estimates}, for $\alpha\in[0,1]$, define
  $\psi(t,\alpha)=\phi\big(t,\alpha x_0+(1-\alpha)y_0\big)$ and note
  $\psi(0,\alpha) = \alpha x_0+(1-\alpha)y_0$ and $\frac{\partial
    \psi}{\partial \alpha}(0,\alpha)= x_0-y_0$.  We then compute:
  \begin{align*}
    \frac{\partial}{\partial t} 
    \frac{\partial}{\partial \alpha} \psi(t,\alpha) &= 
    \frac{\partial}{\partial \alpha}   \frac{\partial}{\partial t} \psi(t,\alpha)  \\
     &= 
    \frac{\partial}{\partial \alpha}   f(t, \psi(t,\alpha) )
    = \frac{\partial f}{\partial x} (t, \psi(t,\alpha) ) 
    \frac{\partial}{\partial \alpha}   \psi(t,\alpha) .
  \end{align*}
  Therefore, $\frac{\partial \psi}{\partial \alpha}(t,\alpha)$ satisfies
  the linear time-varying differential equation $ \frac{\partial}{\partial
    t} \frac{\partial \psi }{\partial \alpha} = \jac{f}(t,\psi)
  \frac{\partial \psi}{\partial \alpha} $. Since the
  system~\eqref{eq:TV-nonlinear} has property~\eqref{InfI},
  Theorem~\ref{thm:coppel} implies
  \begin{align}
    \label{eq:bound-Coppel}
    \verti{\tfrac{\partial \psi}{\partial \alpha}(t,\alpha)} &\leq 
   \verti{\tfrac{\partial \psi}{\partial \alpha}(0,\alpha)}\exp\Big( \int_{0}^t\mu\big( \jac{f}(t,\psi(\tau,\alpha)) \big) d\tau\Big)
    \nonumber \\ &  \leq e^{-ct} \verti{x_0-y_0},
  \end{align}
  where we used $\mu_{\verti{\cdot}}(\jac f(t,x))\le -c$, for every $t\in
  \real_{\ge 0}$ and  $x\in \real^n$. In turn, 
  inequality~\eqref{eq:bound-Coppel} implies
  \begin{align*}
    & \verti{\phi(t,x_0)-\phi(t,y_0)} =\verti{\psi(t,1)-\psi(t,0)} \\ & =
    \verti{\int_{0}^{1}\tfrac{\partial \psi(t,\alpha)}{\partial
        \alpha} d\alpha} \le \int_{0}^{1}\verti{\tfrac{\partial \psi(t,\alpha)}{\partial
        \alpha}}
    d  \alpha \leq e^{-ct} \verti{x_0-y_0}.                            
  \end{align*}
  This completes the proof of part~\ref{p1:estimates}.

  Regarding part~\ref{p2:convergence}, let $\mathcal{P}$ be the orthogonal
  projection onto $\Ker\verti{\cdot}^{\perp}$ and consider the dynamical
  system on $\Ker\verti{\cdot}^{\perp}$:
  \begin{align}\label{eq:aux}
    \dot{y} = f_{\mathcal{P}}(t,y) := \mathcal{P} f(t,y + (I_n-\mathcal{P})\phi(t,x_0)).
  \end{align}
  Note that, for every $v\in \real^{n}$, we have
  $\verti{\mathcal{P}v} =\verti{v}$. This implies that, for $A\in
  \real^{n\times n}$, 
  \begin{align*}
    \mu_{\verti{\cdot}}(\mathcal{P}A) &= \lim_{h\to
    0^{+}}\frac{\verti{\mathcal{P}+h\mathcal{P}^2A}-1}{h} \\ & = \lim_{h\to
    0^{+}}\frac{\verti{\mathcal{P}+h\mathcal{P}A}-1}{h} =\lim_{h\to
    0^{+}}\frac{\verti{I_n+hA}-1}{h}\\ & = \mu_{\verti{\cdot}}(A).
    \end{align*}
  As a result, for every $(t,y)\in \real_{\ge 0}\times \Ker\verti{\cdot}^{\perp}$,
  \begin{align*}
    \mu_{\verti{\cdot}}(Df_{\mathcal{P}}(t,y))  &=
    \mu_{\verti{\cdot}}(\mathcal{P}Df(t,y +
    (I_n-\mathcal{P})\phi(t,x_0))  \\
    & = \mu_{\verti{\cdot}} (Df(t,y +(I_n-\mathcal{P})\phi(t,x_0))) \le -c.
  \end{align*}
  where for the last inequality, we used the fact that
  $\mu_{\verti{\cdot}}(\jac f(t,x))\le -c$, for every $(t,x)\in \real_{\ge
    0}\times \real^n$. Recall that $\verti{\cdot}$ is a norm on
  $\Ker\verti{\cdot}^{\perp}$. Thus, the dynamical system~\eqref{eq:aux} is
  contractive with respect to $\verti{\cdot}$ on
  $\Ker\verti{\cdot}^{\perp}$. Additionally, $(I_n-\mathcal{P})$ is the
  orthogonal projection onto $\Ker\verti{\cdot}$ and therefore, for every
  $t\ge 0$,
  \begin{align*}
    x^* + (I_n-\mathcal{P})\phi(t,x_0) \in x^*+\Ker\verti{\cdot}.
  \end{align*}
  Using~\eqref{ShfI}, for every $t\ge 0$, we obtain
  \begin{align*}
    f(t,x^* + (I_n-\mathcal{P})\phi(t,x_0)) \subseteq \Ker\verti{\cdot}.
  \end{align*}
  This implies that $f_{\mathcal{P}}(t,x^*) = \mathcal{P} f(t,x^* +
  (I_n-\mathcal{P})\phi(t,x_0)) =\vect{0}_n$. As a result, $x^*\in
  \Ker\verti{\cdot}^{\perp}$ is an equilibrium point of the dynamical
  system~\eqref{eq:aux}.  Moreover, one can see that $t\mapsto
  \mathcal{P} \phi(t,x_0)$ is another trajectory of the dynamical
  system~\eqref{eq:aux}. Since the dynamical system~\eqref{eq:aux} is
  contractive on $\Ker\verti{\cdot}^{\perp}$,
  \begin{align*}
    \verti{\mathcal{P}\phi(t,x_0) -
      x^*} & \le e^{-ct}\verti{\mathcal{P}\phi(0,x_0)
      - x^*} \\ & = e^{-ct}\verti{x_0- x^*},
  \end{align*}
  where the last equality hold because
  $\Ker\verti{\cdot} = \Img(I_n-\mathcal{P})$. Moreover,
  \begin{align*}
    \lim_{t\to\infty}\verti{\mathcal{P}\phi(t,x_0) - x^*}
    \le \lim_{t\to\infty}  e^{-ct}\verti{x_0- x^*} = 0. 
  \end{align*}
  Note that $\mathcal{P}\phi(t,x_0) - x^*\in \Ker\verti{\cdot}^{\perp}$ and
  $\verti{\cdot}$ is a norm on $\Ker\verti{\cdot}^{\perp}$. This implies
  that $\lim_{t\to\infty}\mathcal{P}\phi(t,x_0) =x^*$, or equivalently,
  $\lim_{t\to\infty}\phi(t,x_0)\in
  x^*+\Ker\verti{\cdot}$ with exponential convergence rate $c$.

  Regarding part~\ref{p3:f}, note that the map $t\mapsto f(\phi(t,x_0))$
  satisfies $\tfrac{d}{dt} f(\phi(t,x_0)) = \jac f(\phi(t,x_0))
  f(\phi(t,x_0))$. Since~\eqref{InfI} holds, Theorem~\ref{thm:coppel} implies
  \begin{align*}
    &\verti{f(\phi(t,x_0))} \\ & \le
    \exp\left(\int_{0}^{t}\mu_{\verti{\cdot}}\big(\jac
    f(\phi(\tau,x_0))\bigm) d\tau\right) \verti{f(\phi(0,x_0))}\\ & \le e^{-ct} \verti{f(x_0)},
  \end{align*}
where we used $\mu_{\verti{\cdot}}(\jac f(t,x))\le -c$, for every $t\in
\real_{\ge 0}$ and $x\in C$.

 Regarding part~\ref{p4:invariantsubspace}, consider the dynamical system
 on $\Ker\verti{\cdot}^{\perp}$
 \begin{align}\label{eq:auxD}
   \dot{y} = \mathcal{P}f(y).
 \end{align} 
 Note that $\mu_{\verti{\cdot}}(\mathcal{P}Df(y)) = \mu_{\verti{\cdot}}(Df(y)) \le
 -c$, for every $y\in \Ker\verti{\cdot}^{\perp}$ and $\verti{\cdot}$
 is a norm on $\Ker\verti{\cdot}^{\perp}$. Therefore, the dynamical
 system~\eqref{eq:auxD} is contracting on $\Ker\verti{\cdot}^{\perp}$ with
 respect to $\verti{\cdot}$. Since~\eqref{eq:auxD} is time-invariant,
 it has a unique globally stable equilibrium point $x^*\in
 \Ker\verti{\cdot}^{\perp}$. Thus, we have $\mathcal{P}f(x^*) =
 \vect{0}_n$ or equivalently $f(x^*)\in \Ker\verti{\cdot}$. For every $x_0\in \real^n$, recall that $t\mapsto \phi(t,x_0)$ is the 
 trajectory of the dynamical system~\eqref{eq:TV-nonlinear} starting
 form $x_0$. Therefore, by part~\ref{p3:f},
\begin{align*}
\verti{f(\phi(t,x^*))} \le e^{-ct}\verti{f(x^*)} = 0.
\end{align*}
As a result $f(\phi(t,x^*))\in \Ker\verti{\cdot}$, for every $t\ge
0$. This implies that the trajectory $t\mapsto \phi(t,x^*)$ remains in
$x^*+\Ker\verti{\cdot}$, for every $t\ge 0$. Using part~\ref{p1:estimates}, for every $x_0\in
 x^*+\Ker\verti{\cdot}$ and every $t\in \real_{\ge 0}$,
 \begin{align*}
    \verti{\phi(t,x_0)-\phi(t,x^*)} \le e^{-ct} \verti{x_0-x^*} = 0,
 \end{align*}
 where the last equality holds because $x_0-x^*\in \Ker\verti{\cdot}$
 and thus $ \verti{x_0-x^*} = 0$. This means that, for every $x_0\in
 x^*+\Ker\verti{\cdot}$ and every $t\in \real_{\ge 0}$, we have $\phi(t,x_0)\in
 x^*+\Ker\verti{\cdot}$. As a consequence, property~\eqref{ShfI}
 holds for $x^*$.
 \end{proof}

\begin{remark}[Comparison with literature]\label{rem:compare}
  Theorem~\ref{thm:partialcontraction}\ref{p1:estimates} is related to
  horizontal contraction theory, as developed in~\cite{FF-RS:14}.
  Compared to horizontal contraction theory that can examine
  convergence to general submanifolds, semi-contraction theory studies
  convergence to vector subspaces. On the other hand, the notion of semi-measure leads to a
  {\color{black}Coppel's inequality for semi-norms}, to a less restrictive
  invariance condition for semi-contractivity, and to sharp results about
  convergence rates. In short, semi-contraction theory offers a more
  comprehensive treatment with a smaller domain of applicability. We
  postpone a comprehensive comparison to Appendix~\ref{app:semi-horizontal}.
  
  Theorem~\ref{thm:partialcontraction}\ref{p2:convergence} is related to
  partial contraction
  theory~\cite{QCP-JJS:07,MdB-DF-GR-FS:16,PCV-SJ-FB:19r}.  Specifically,
  Theorem~\ref{thm:partialcontraction}\ref{p2:convergence} is more general
  than~\cite[Theorem~1]{QCP-JJS:07}, whereby only the $\ell_2$-norm is
  considered, and \cite{MdB-DF-GR-FS:16,PCV-SJ-FB:19r}, whereby partial
  contraction is defined using an orthonormal set of vectors.  In contrast,
  our statement and proof of
  Theorem~\ref{thm:partialcontraction}\ref{p2:convergence} are geometric
  and coordinate independent.

  To the best of our knowledge, the results in
  Theorem~\ref{thm:partialcontraction}\ref{p3:f} and
  Theorem~\ref{thm:partialcontraction}\ref{p4:invariantsubspace} are novel. \oprocend
\end{remark}

\begin{remark}[Computation of suitable semi-norms]\label{rem:computation}
   An important question is how to find a semi-norm with respect to which a
  given system is semi-contracting. In some cases, one can exploit the
  system structure to guess a candidate semi-norm and then verify
  semi-contractivity; several such examples are in
  Section~\ref{sec:app}. Alternatively, one may resort to computational
  methods. In classical contraction theory, a well-established and
  efficient computational method for polynomial systems is sum of square
  (SOS) programming~\cite{EA-PAP-JJES:08}. This approach can be easily
  extended to treat weighted $\ell_2$-semi-norms.  Consider the dynamical
  system~\eqref{eq:TV-nonlinear} with $f$ polynomial in $(t,x)$ and a
  vector subspace $\mcS\subseteq \real^n$. Using
  Theorem~\ref{lem:computational}\ref{p5:2-norm}, one can compute a
  suitable semi-norm of the form $\verti{\cdot}_{2,R}$ with $\Ker
  \verti{\cdot}_{2,R} =\mcS$ by solving the following SOS program in $P\in
  \real^{n\times n}$:
  \begin{align}\label{prob:checking_demidovich-SOS}
    &y^{\top}Q^{\top}(-P \jac{f}(t,x) - (\jac{f}(t,x))^{\top}P + 2c P )Qy \in \Sigma(t,x,y),\nonumber\\
    &y^{\top}P y \in \Sigma(t,x,y),\nonumber\\
    &Pv_1 = Pv_2 = \cdots = Pv_{n-k} = \vect{0}_n,
  \end{align}
  where $y\in \real^n$ is an intermediate variable, $\Sigma(t,x,y)$ is the
  ring of sum of square polynomials, and $\{v_1,\ldots,v_{n}\}$ is an
  orthonormal basis for $\real^n$ with $\{v_1,\ldots,v_{n-k}\}$ a basis
  for  $\mcS$ and $Q = \begin{bmatrix} \vect{0}_{n\times (n-k)},v_{n-k+1},\ldots,v_n\end{bmatrix}\in
  \real^{n\times n}$.\oprocend
 \end{remark}
 
It should be noted that
Theorem~\ref{thm:partialcontraction}\ref{p2:convergence} does not
state that the asymptotic behavior of generic trajectories is the same
as the asymptotic behavior of the trajectories in the invariant set $x^*+\Ker {\verti{\cdot}}$. The following example elaborates
on this aspect.  }

 \begin{example}\label{ex:semi-contractive}
   The dynamical system on $\real^2$
   \begin{align}
     \begin{split}\label{eq:what-a-lame-system}
     \dot{x}_1 = - x_1,\\
     \dot{x}_2 = x_1x_2^2
     \end{split}
   \end{align}
   is semi-contracting with respect to the semi-norm $\verti{\cdot}$
   defined by $\verti{(x_1,x_2)^{\top}} = |x_1|$, for every
   $(x_1,x_2)^{\top}\in \real^2$. {\color{black}Additionally, $\Ker\verti{\cdot} = \setdef{(0,x_2)}{x_2\in\real}$ is shifted-invariant under~\eqref{eq:what-a-lame-system}.} Thus, by
   Theorem~\ref{thm:partialcontraction}\ref{p2:convergence}, every
   trajectory of the system~\eqref{eq:what-a-lame-system} converges to
   the set $\Ker\verti{\cdot}= \setdef{(0,x_2)}{x_2\in\real}$. Moreover, $\mcS$ is
   the set of equilibrium points for~\eqref{eq:what-a-lame-system}. On the
   other hand, the curve $t\mapsto (x_1(t),x_2(t)):=(e^{-t},e^{t})$ is a
   trajectory of~\eqref{eq:what-a-lame-system} such that
   $\lim_{t\to\infty}x_2(t)=\infty$. Therefore, this trajectory does not
   converge to any equilibrium point. 
  \end{example}

\section{Weakly contracting systems}\label{sec:weak}

We here generalize contraction theory by providing a comprehensive
treatment to weakly contracting systems.

\begin{definition}[Weakly contracting systems]
  Let $C\subseteq \real^n$ be a convex set, $\|\cdot\|$ be a norm on
  $\real^n$, and $\mu$ be its associated matrix measure.  The time-varying
  dynamical system~\eqref{eq:TV-nonlinear} is \emph{weakly contracting} on
  $C$ with respect to $\|\cdot\|$, if
  \begin{align}\label{eq:nonexpansive}
    \mu(\jac{f} (t,x)) \le 0,\quad\mbox{for all } t\in \real_{\ge 0}, x\in C.
  \end{align}
\end{definition}
We aim to study the qualitative behaviors of weakly contracting
systems. For the rest of this section, we restrict ourselves to the
time-invariant dynamical systems:
\begin{align}\label{eq:nonlinear}
  \dot{x} = f(x), 
\end{align}
on convex invariant sets and we assume that $x\mapsto f(x)$ is
twice-differentiable. We start with a useful and novel lemma. 

\begin{lemma}[Weak contraction and bounded trajectory imply invariant set]
  \label{lem:weak-implies-invariant}
  Consider the time-invariant dynamical system~\eqref{eq:nonlinear} with a
  closed convex invariant set $C$. Suppose that $\|\cdot\|$ is a norm on
  $\real^n$, $f$ is continuously differentiable and weakly contracting with
  respect to $\|\cdot\|$, and $t\mapsto x(t)$ is a bounded
  trajectory of the system in $C$. Then the system~\eqref{eq:nonlinear} has
  a compact convex invariant set $W\subseteq C$.
\end{lemma}


Next we establish a useful qualitative dichotomy.

\begin{theorem}[Dichotomy for qualitative behavior of trajectories of
  weakly contracting systems]
  \label{thm:weak-dichotomy}
 Consider the time-invariant dynamical system~\eqref{eq:nonlinear} with a
 convex invariant set $C\subseteq \real^n$. Let $\|\cdot\|$ be a norm on
 $\real^n$ with the associated matrix measure $\mu$. Suppose that the
 system~\eqref{eq:nonlinear} is weakly contracting on $C$ with respect to
 $\|\cdot\|$. Then either of the following exclusive conditions hold:
 \begin{enumerate}
 \item\label{thmwd:p1} the dynamical system~\eqref{eq:nonlinear} has at
   least one stable equilibrium point $x^*\in C$ and every trajectory
   of~\eqref{eq:nonlinear} starting in $C$ is bounded; or
   
 \item\label{thmwd:p2} the dynamical system~\eqref{eq:nonlinear} has no
   equilibrium point in $C$ and every trajectory of~\eqref{eq:nonlinear}
   starting in $C$ is unbounded.
 \end{enumerate}

 Moreover, under condition~\ref{thmwd:p1}, the following statement hold:
 \begin{enumerate}\setcounter{enumi}{2}
 \item\label{thmwd:p4.5} if $\jac{f}(x^*)$ is Hurwitz, then $x^*$ is
   {\color{black}the unique globally exponentially stable equilibrium point
     in $C$ with convergence rate $-\alpha(\jac{f}(x^*))$}.
 \end{enumerate}
\end{theorem}

\begin{proof}
  Suppose that the equation $f(x) = \vect{0}_n$ has at least one
  solution $x^*$ in $C$. Let $t\mapsto x(t)$ be a trajectory of the
  system. Since the system~\eqref{eq:nonlinear} is weakly contracting, we have
  \begin{align*}
    \|x(t)-x^*\| \le \|x(0) - x^*\|,\quad\forall t\ge 0. 
  \end{align*}
  By setting $M = \|x(0) - x^*\|$ and using triangle inequality, we see
  that $\|x(t)\|\le \|x^*\| + M$, for every $t\ge 0$. This implies that
  $t\mapsto x(t)$ is bounded and thus statement~\ref{thmwd:p1} holds. Now
  suppose that the algebraic equation $f(x) = \vect{0}_n$ does not have any
  solution in $C$. We now need to show that every trajectory
  of~\eqref{eq:nonlinear} is unbounded. By contradiction, assume $t\mapsto
  x(t)$ is a bounded trajectory of~\eqref{eq:nonlinear}. In
  Lemma~\ref{lem:weak-implies-invariant} we establish that there exists a
  compact convex invariant set $W\subseteq C$ for the dynamical
  system~\eqref{eq:nonlinear}. Therefore, by Yorke
  Theorem~\cite[Lemma~4.1]{AL-JAY:76}, the system~\eqref{eq:nonlinear} has
  an equilibrium point inside $W\subseteq C$. This is in contradiction with
  the assumption that $f(x) = \vect{0}_n$ has no solution inside
  $C$. Therefore, every trajectory of~\eqref{eq:nonlinear} is
  unbounded. This completes the proof of the dichotomy.

  Regarding part~\ref{thmwd:p4.5}, $\jac{f}(x^*)$ being Hurwitz implies
  that $x^*$ is a locally exponentially stable equilibrium point
  for~\eqref{eq:nonlinear}. Therefore, there exists $\epsilon, T > 0$ such
  that $\overline{B}_{\|\cdot\|}(x^*,\epsilon)$ is in the region of
  attraction of the equilibrium point $x^*$ and, for every $z\in
  \overline{B}_{\|\cdot\|}(x^*,\epsilon)$, we have $\phi(T,z) \in
  \overline{B}_{\|\cdot\|}(x^*,\epsilon/2)$. Let $t\mapsto x(t)$ denote a
  trajectory of the dynamical system. Assume that $y\in \partial
  B_{\|\cdot\|}(x^*,\epsilon)$ is a point on the straight line connecting
  $x(0)$ to the unique equilibrium point $x^*$. Then we have
  \begin{align*}
    \left\|x(T) - x^*\right\| &\le  \left\|x(T) -
    \phi(T,y)\right\| + \left\|\phi(T,y)- x^*\right\|
    \\& \le
    \left\|x(0)-y\right\| + \epsilon/2 =\|x(0)- x^*\|-\epsilon/2,
  \end{align*}
  where the last equality holds because $x^*$, $y$, and $x(0)$ are on the
  same straight line. Therefore, after time $T$, $t\mapsto \left\|x(t) -
  x^*\right\|$ decreases by $\epsilon/2$. As a result, there exists a
  finite time $T_{\inf}$ such that, for every $t\ge T_{\inf}$, we have
  $x(t)\in \overline{B}_{\|\cdot\|}(x^*,\epsilon)$. Since
  $\overline{B}_{\|\cdot\|}(x^*,\epsilon)$ is in the region of attraction
  of $x^*$ the trajectory $t\mapsto x(t)$ converges to $x^*$. {\color{black}Since
  $\jac{f}(x^*)$ is Hurwitz, the convergence rate of~\eqref{eq:nonlinear}
  is equal to the convergence rate of its linearization around $x^*$
  which is $-\alpha(\jac{f}(x^*))$.}
\end{proof}
{\color{black}
  \begin{remark}[Comparison with literature]    
    To the best of our knowledge, the dichotomy for weakly contracting
    system (Theorem~\ref{thm:weak-dichotomy}\ref{thmwd:p1}
    and~\ref{thmwd:p2}) is novel. Regarding
    Theorem~\ref{thm:weak-dichotomy}\ref{thmwd:p4.5}:
    \begin{enumerate}
    \item A special case for weakly contracting system with respect to
      $\ell_1$-norm is proved in~\cite[Lemma 6]{EL-GC-KS:14};
    \item A weaker version for when $x^*$ is locally asymptotically (but
      not necessarily exponentially) stable is proved in~\cite[Lemma
        III.1]{PCV-SJ-FB:19r}.
    \item A stronger version for when $\mu(\jac{f}(x^*))<0$ is due
      to~\cite[Corollary 8]{SC:19}. \oprocend
    \end{enumerate}
\end{remark}}

The dichotomy in Theorem~\eqref{thm:weak-dichotomy} 
characterizes the qualitative asymptotic behavior of a weakly contracting system when it
has no equilibrium point. However, for weakly contracting systems with
at least one equilibrium point, this theorem only guarantees
boundedness of the trajectories. The following theorem shows that, for
some classes of norms, weak contractivity implies convergence to equilibrium points.

 \begin{theorem}[Weak contraction and convergence to equilibria]\label{thm:1-inf}
   Consider the time-invariant dynamical system~\eqref{eq:nonlinear} with a
   convex invariant set $C\subseteq \real^n$. Suppose that the vector field
   $f$ is differentiable and piecewise real analytic with an equilibrium
   point $x^*\in C$. Suppose that there exists $p\in\{1,\infty\}$ and an
   invertible matrix $Q\in \real^{n\times n}$ such that
   $\mu_{p,Q}(\jac{f}(x))\le 0$, for every $x\in C$, i.e., the system is
   weakly contracting with respect to $\|\cdot\|_{p,Q}$ on $C$. Then every
   trajectory of~\eqref{eq:nonlinear} starting in $C$ converges to an
   equilibrium point.
\end{theorem}


\begin{proof} We prove the theorem for $p=1$. The proof for
  $p=\infty$ is similar and we omit it. Suppose that the vector field $f$
  is piecewise real analytic and the system~\eqref{eq:nonlinear} is
  weakly contracting with respect to $Q$-weighted $\ell_1$-norm. By
  Theorem~\ref{thm:weak-dichotomy}\ref{thmwd:p1} every trajectory of the
  system~\eqref{eq:nonlinear} is bounded. Now we use the LaSalle Invariance
  Principle~\cite[Theorem 15.7]{FB:20} for the function
  $V(x)=\|f(x)\|_{1,Q}$ on the invariant convex set $C$. 
  For every $c>0$, $V$ is continuous and
  $V^{-1}(c)$ is closed. Moreover, for every trajectory $t\mapsto x(t)$ starting inside
  $V^{-1}(c)$, we have $\tfrac{d}{dt} f(x(t)) = \jac f(x(t))
  f(x(t))$. Using Theorem~\ref{thm:coppel},
    \begin{align*}
      \|f(x(t))\|_{1,Q} &\le \exp\left(\int_{0}^{t}\mu_{1,Q}(\jac
      f(x(\tau)))d\tau\right) \|f(x(0))\|_{1,Q}\\ & \le \|f(x(0))\|_{1,Q}
      \le c,
    \end{align*}
    where the second inequality holds because system~\eqref{eq:nonlinear}
    is weakly contracting with respect to $Q$-weighted $\ell_1$-norm and thus
    $\mu_{1,Q}(\jac f(x))\le 0$, for every $x\in C$ and the last inequality
    holds because $x(0) \in V^{-1}(c)$. Thus, $V^{-1}(c)$ is a closed invariant
    set for~\eqref{eq:nonlinear}.  Therefore, by the LaSalle Invariance
    Principle, the largest invariant set $M$ inside the set
  \begin{equation*}
     \setdef{x\in C}{ \mathcal{L}_f V (x) = 0}\cap V^{-1}(c)
  \end{equation*}
  is nonempty and every trajectory of~\eqref{eq:nonlinear} converges to
  $M$. Let $t\mapsto y(t)$ denote a trajectory in the set $M$ (and
  therefore in the set $C$). Our goal is to show that $t\mapsto y(t)$ is an
  equilibrium point. Since the vector field
  $f$ is piecewise real analytic, there exists a partition
  $\{\Sigma_j\}_{j=1}^{m}$ of $C$ such that $f$ is real analytic on
  $\mathrm{int}(\Sigma_j)$, for every $j\in \{1,\ldots,m\}$. Now consider
  the trajectory $t\mapsto y(t)$. It is clear that, there exists $k\in
  \{1,\ldots,m\}$ and $T>0$ such that $y(t)\in \mathrm{cl}(\Sigma_k)$, for
  every $t\in [0,T]$. Since $f$ is real analytic on
  $\mathrm{int}(\Sigma_k)$, there exists a real analytic vector field $g:
  C\to \real^n$ such that $g(x) = f(x)$, for every $x\in
  \mathrm{cl}(\Sigma_{k})$. This implies that, for every $t\in [0,T]$, the
  curve $t\mapsto y(t)$ is a solution of the dynamical system
    \begin{align*}
      \dot{y}(t) = f(y(t)) = g(y(t)). 
      \end{align*}
Therefore, the curve $t\mapsto y(t)$ is real analytic for every $t\in
[0,T]$.  Note that, for every $t\in [0,T]$,  
   \begin{align*}
    0 &= \frac{d}{dt} V (y(t)) = \sum_{i=1}^{n} \sign\!\Big( (Qf)_i(y(t)) \Big)
    (Q\dot{f})_i(y(t)) \\ & = \sum_{i=1}^{n} \sign\!\Big( (Qg)_i(y(t)) \Big)
    (Q\dot{g})_i(y(t)).
  \end{align*}
 Since $g$ is real analytic and $t\mapsto y(t)$ is real analytic on
 $[0,T]$, then $t\mapsto (Qg)_i(y(t))$ is real analytic, for $i\in
 \{1,\ldots,n\}$. Therefore, either of the following conditions hold:
  \begin{enumerate}
  \item $(Qg)_i(y(t))\ne 0$, for every $t\in [0,T]$, or
  \item $(Qg)_i(y(t))= 0$, for every $t\in [0,T]$.
  \end{enumerate}
  Since $t\mapsto (Qg)_i(y(t))$ is continuous, for
 every $i\in \{1,\ldots,n\}$, this implies that $\sign((Qg)_i(y(t))) = \sign((Qg)_i(y(0)))$, for every
  $t\in[0,T]$ and every $i\in \{1,\ldots,n\}$. Define $\mathbf{w}\in
  \real^n$ by
  \[
  \mathbf{w} = \sign(Qg(y(0))).
  \]
  Thus, we have $\mathcal{L}_f V (y(t)) = \mathbf{w}^{\top}Q\dot{g}(y(t)) = 0$, for every $t\in [0,T]$.
  By integrating this condition, we get
  \[
  \mathbf{w}^{\top}Qg(y(t)) = \beta,\quad\forall t\in [0,T],
  \]
  for some constant $\beta\ge 0$. We first show that $\eta=0$. Note that $g(y(t))=f(y(t)) =
  \dot{y}(t)$. Therefore, we have $\mathbf{w}^{\top}Q\dot{y}(t) = \beta$, for
  every $t\in [t_1,t_2]$. Integrating with respect to time, we get
  \[
  \mathbf{w}^{\top}Q y(t) = \beta t + \eta, \quad\forall  t\in [0,T].  
  \]
  for some constant $\eta\in \real$. Since every trajectory of~\eqref{eq:nonlinear}
  starting in $C$ is bounded, we have $\beta = 0$. Now, note that, for every
  $t\in [0,T]$,
  \begin{align*}
    \|f(y(t))\|_{1,Q} & =
    \sum_{i=1}^{n}\sign((Qf)_i(y(t))(Qf)_i(y(t)) \\ & =
    \mathbf{w}^{\top}Qg(y(t)) = 0 .
  \end{align*}
  This implies that $f(y(t)) = \vect{0}_n$, for every $t\in [0,T]$. Since
  $t\mapsto y(t)$ is a trajectory of~\eqref{eq:nonlinear} and it is
  continuous, we have $f(y(t)) = \vect{0}_n$, for every $t\ge 0$. This
  implies that every trajectory inside $M$ is an equilibrium
  point. Therefore, every trajectory of~\eqref{eq:nonlinear} starting in
  $C$ converges to the set of equilibrium points.
\end{proof}


\begin{remark}[{\color{black}Convergence to equilibrium points}]
  \begin{enumerate}
     

  \item We believe that Theorem~\ref{thm:1-inf} can be generalized to
    polyhedral norms (see~\cite[Definition 5.5.2]{RAH-CRJ:12} for
    definition of polyhedral norms); we omit this generalization in
    the interest of brevity.

  \item Theorem~\ref{thm:1-inf} does not necessarily hold for systems that
    are weakly contracting with respect to $\ell_p$-norms for
    $p\not\in\{1,\infty\}$. The following example illustrates this fact for
    the $\ell_2$-norm. \oprocend
  \end{enumerate} 
\end{remark}

\begin{example}\label{ex:weak-contractive}
  The dynamical system on $\real^2$
  \begin{equation}\label{eq:circles}
  \begin{split}
    &\dot{x}_1 = x_2,\\ &\dot{x}_2 = -x_1,
  \end{split}
  \end{equation}
  has a unique equilibrium point $\vect{0}_2$ and its vector field $f$ is
  piecewise real analytic. Let $\mu_2$ denote the matrix measure with
  respect to the $\ell_2$-norm on $\real^2$.  Then
  \begin{align*}
    \mu_2(\jac{f}(x)) = \lambda_{\max}\!\left(\frac{\jac{f}(x) +
      \jac{f}^{\top}(x)}{2}\right)
    =
    \lambda_{\max}\left(\vect{0}_{2\times
      2}\right) = 0,
  \end{align*}
  so that system~\eqref{eq:circles} is weakly contracting with respect
  to~$\|\cdot\|_2$. However, $t\mapsto (\sin(t),\cos(t))$ is a periodic
  trajectory of~\eqref{eq:circles} which does not converge to the
  equilibrium point. \oprocend
\end{example}

\section{Doubly-contracting systems}\label{sec:double}

Examples~\ref{ex:semi-contractive} and~\ref{ex:weak-contractive} illustrate
that, {\color{black} for time-invariant systems}, semi-contractivity or
weak-contractivity alone do not guarantee the convergence of trajectories
to equilibrium points. Here we show that, for time-invariant systems, a
combination of these two properties can ensure that every trajectory
converges to an equilibrium point.

\begin{theorem}[Double contraction and convergence to equilibria]\label{thm:mixed}
Consider the time-invariant dynamical system~\eqref{eq:nonlinear} with a
convex invariant set $C\subseteq \real^n$. Suppose that $\mcS\subseteq
\real^n$ is a vector subspace {\color{black}consisting of} equilibrium points
of~\eqref{eq:nonlinear}. Assume that there exist
\begin{enumerate}[label*=(A\arabic*)]
\item\label{assu:weak} a norm $\|\cdot\|$ such that~\eqref{eq:nonlinear} is
  weakly contracting with respect to $\|\cdot\|$ on $C$;  and
\item\label{assu:partial} a semi-norm $\verti{\cdot}$ with
  $\Ker\verti{\cdot}=\mcS$ such that~\eqref{eq:nonlinear} is
  semi-contracting with respect to $\verti{\cdot}$ on $C$. 
\end{enumerate}
Then, for every trajectory $t\mapsto x(t)$ of~\eqref{eq:nonlinear} starting
in $C$, there exists $x^*\in \mcS$ such that $\lim_{t\to\infty}x(t)=x^*$
with exponential convergence rate $-\alpha_{\mcS^{\perp}}(\jac
f(x^*))$.
\end{theorem}


We refer to systems satisfying~\ref{assu:weak} and~\ref{assu:partial} as
\emph{doubly-contracting} in the sense that their trajectories satisfy a
weak contractivity and a semi-contractivity property:
{\color{black}\begin{align*}
  \norm{\phi(t,x_0) - \phi(t,y_0)} &\le    \norm{x_0-y_0},
  \\
  \verti{\phi(t,x_0)} &\le    e^{-ct}\verti{x_0},
\end{align*}}
for every $x_0,y_0\in C$, every $t\in \real_{\ge 0}$, and some $c \in \real_{\ge 0}$.

\begin{proof}
  Let $t\mapsto x(t)$ denote a trajectory of~\eqref{eq:nonlinear} starting
  from $x(0)\in C$.  Let $\mathcal{P}$ be the orthogonal projection onto
  the subspace $\mcS^{\perp}=\Ker\verti{\cdot}^{\perp}$. Note that, for
  every $t\in \real_{\ge 0}$, the orthogonal projection of $x(t)$ onto
  $\mcS=\Ker\verti{\cdot}$ is given by $(I_n-\mathcal{P})x(t)$ and it is an
  equilibrium points for~\eqref{eq:nonlinear}. Moreover, by
  Assumption~\ref{assu:weak}, system~\eqref{eq:nonlinear} is
  weakly contracting with respect to the norm $\|\cdot\|$. This implies
  that, for every $t\in \real_{\ge 0}$ and every $s\ge t$, the point $x(s)$
  remains inside the closed ball
  $\overline{B}_{\|\cdot\|}((I_n-\mathcal{P})x(t),
  \|\mathcal{P}x(t)\|)$. Therefore, for every $t\ge 0$, the point $x(t)$ is
  inside the set $D_t$ defined by
  \begin{align*}
    D_t = \mathrm{cl}\big(\bigcap_{\tau\in [0,t]} \overline{B}_{\|\cdot\|}((I_n-\mathcal{P})x(\tau),
    \|\mathcal{P}x(\tau)\|)\big). 
  \end{align*}
  It is easy to see that, for $s\ge t$, we have $D_{s}\subseteq D_t$. This
  implies that the family $\{D_t\}_{t\in [0,\infty)}$ is a nested family of
    closed subsets of $\real^n$ such that $\mathrm{diam}(D_t)\le
    \|\mathcal{P}x(t)\|$, for every $t\in [0,\infty)$. On the other hand, for every $t\in
    [0,\infty)$, the set $D_t$ is a closed convex invariant set for the
    system~\eqref{eq:nonlinear}. {\color{black} Moreover, $\Ker\verti{\cdot}=\mcS$ consists
    of equilibrium points of~\eqref{eq:nonlinear}. This implies that,
    $\Ker\verti{\cdot}$ is shifted-invariant under the
    system~\eqref{eq:nonlinear}.} Using Assumption~\ref{assu:partial} and Theorem~\ref{thm:partialcontraction}\ref{p2:convergence}, the
    trajectory $t\mapsto x(t)$ converges to the subspace $\mcS$.  This
    means that
    $\lim_{t\to\infty}\mathrm{diam}(D_t)=\lim_{t\to\infty}\|\mathcal{P}x(t)\|=0$. Thus,
    by the Cantor Intersection Theorem~\cite[Lemma 48.3]{JM:00}, there
    exists $x^*\in C$ such that $\bigcap_{t\in [0,\infty)} D_t =
      \{x^*\}$. We show that $\lim_{t\to \infty} x(t) = x^*$. Note that
      $x^*,x(t)\in D_t$, for every $t\in \real_{\ge 0}$. This implies that
      $\|x(t) - x^*\|\le \mathrm{diam}(D_t)$.  This in turn means that
      $\lim_{t\to\infty} \|x(t) - x^*\|=0$ and $t\mapsto x(t)$ converges to
      $x^*$. On the other hand, the trajectory $t\mapsto x(t)$ converges to
      the subspace $\mcS$. Therefore, $x^*\in \mcS$ and {\color{black}it is an
      equilibrium point.} Regarding the convergence rate, by
      Theorem~\ref{thm:abscissa}, we have
      \begin{align*}
        &\alpha_{\mcS^{\perp}}(\jac
        f(x^*))\\ & =\inf\setdef{\mu_{\verti{\cdot}}(\jac
        f(x^*))}{\verti{\cdot}\mbox{ a semi-norm, } \Ker\verti{\cdot}=\mcS}.
      \end{align*}
      First note that system~\eqref{eq:nonlinear} is semi-contracting
      with respect to $\verti{\cdot}$ on $C$. This implies that
      $\alpha_{\mcS^{\perp}}(\jac f(x^*)) \le \mu_{\verti{\cdot}}(\jac
      f(x^*)) < 0$. Let $\epsilon>0$ be such that
      $\epsilon\le \tfrac{1}{3}\left|\alpha_{\mcS^{\perp}}(\jac
        f(x^*))\right|$. There exists a semi-norm
      $\verti{\cdot}_{\epsilon}$ such that
      $\mu_{\verti{\cdot}_{\epsilon}}(\jac f(x^*)) \le
      \alpha_{\mcS^{\perp}}(\jac f(x^*)) + \epsilon < 0$. Consider the
      family of closed convex invariant sets $\{D_t\}_{t\ge 0}$. Since
      $\lim_{t\to\infty}\mathrm{diam}(D_t)= 0$, and $f$ is twice
      differentiable, there exists $t_{\epsilon}\in \real_{\ge 0}$
      such that
      $\mu_{\verti{\cdot}_{\epsilon}}(\jac f(x)) \le
      \alpha_{\mcS^{\perp}}(\jac f(x^*))+ 2\epsilon < 0$, for every
      $x\in D_{t_{\epsilon}}$. Therefore, using
      Theorem~\ref{thm:partialcontraction}\ref{p2:convergence}, the
      trajectory $t\mapsto x(t)$ converges to the subspace $\mcS$ with
      the convergence rate
      $-\alpha_{\mcS^{\perp}}(\jac f(x^*))-
        2\epsilon$. Since $\epsilon$ can be chosen arbitrarily
      small, the trajectory $t\mapsto x(t)$ converges to the subspace
      $\mcS$ with the rate
      $-\alpha_{\mcS^{\perp}}(\jac f(x^*))$. {\color{black}
        As a result, there exists $M>0$ such that, for every $t\ge 0$,
        \begin{align*}
          \|\mathcal{P}x(t)\|\le
          M e^{\alpha_{\mcS^{\perp}}(\jac f(x^*))t}. 
        \end{align*}
        Moreover, we know that $\|x(t)-x^*\|\le
        \mathrm{diam}(D_t)\le \|\mathcal{P}x(t)\|$, for every $t\ge
        0$. This implies that, for every $t\ge 0$, 
        \begin{align*}
          \|x(t)-x^*\|\le M e^{\alpha_{\mcS^{\perp}}(\jac f(x^*))t}.
        \end{align*}
        Thus, $\lim_{t\to\infty}x(t) = x^*$ with rate $-\alpha_{\mcS^{\perp}}(\jac
        f(x^*))$.}
\end{proof}

\begin{remark}
  Theorem~\ref{thm:mixed} provides a convergence rate for trajectories of
  the system which depends only on the converging point and is independent
  of both the norm $\|\cdot\|$ in Assumption~\ref{assu:weak} and the
  semi-norm $\verti{\cdot}$ in Assumption~\ref{assu:partial}. This theorem
  also generalizes the classical results in contraction theory where the
  convergence rate depends on the norm. \oprocend
\end{remark}

\section{Application to network systems}\label{sec:app}
In this section, we apply our results to example systems.  We show that (i)
affine averaging systems and affine flow systems are doubly-contracting,
(ii) distributed primal-dual dynamics is weakly contracting, and (iii)
networks of {\color{black}diffusively-coupled dynamical systems} with strong coupling are
semi-contracting.

\subsection{Affine averaging and flow systems are doubly-contracting}

\begin{theorem}[Affine averaging system]\label{thm:affine-avg}
  Let $L$ be the Laplacian of a weighted digraph with a globally reachable
  node.  Let $v$ be the dominant left eigenvector of $L$ satisfying
  $\vect{1}_n^{\top}v=1$. Let $b\in\real^n$. Then the \emph{affine
    averaging system}
  \begin{align}\label{eq:ave}
    \dot{x} = - L x + b
  \end{align}
  \begin{enumerate}
    \item\label{p0:contractive} is weakly contracting with respect to
      $\|\cdot\|_{\infty}$ and semi-contracting with respect to
      $\|\cdot\|_{\infty,R_\epsilon}$, where $R_\epsilon$ is defined in
      Remark~\ref{remark:Reps} as function of a sufficiently small
      $\epsilon$,
      
  \item\label{p1:unbounded} if $v^{\top}b\neq0$, then every trajectory is
    unbounded, 
    
  \item\label{p2:bounded} if $v^{\top}b=0$, then the trajectory starting
    from $x(0)=x_0$ converges to the equilibrium point $L^{\dagger}b +
    (v^{\top}x_0)\vect{1}_n$ with exponential rate $-\alpha_{\mathrm{ess}}(-L)$.
    
   \end{enumerate}
 \end{theorem}


\begin{proof}
  Regarding part~\ref{p0:contractive}, for $f(x):=-Lx + b$, note
  $\mu_{\infty}(\jac f(x)) =\mu_{\infty}(-L)= 0$. Thus, the
  system~\eqref{eq:ave} is weakly contracting with respect to
  $\ell_\infty$-norm. Regarding part~\ref{p1:unbounded}, if
  $v^{\top}b\neq0$, then there does not exist $x\in \real^n$ such that $-Lx
  + b = 0$. Thus, Theorem~\ref{thm:weak-dichotomy}\ref{thmwd:p1} implies
  every trajectory of~\eqref{eq:ave} is unbounded.  Regarding
  part~\ref{p2:bounded}, since $G$ has a reachable node, $L$ has a simple
  eigenvalue $0$ and its other eigenvalues have positive real parts. Thus,
  if $v^{\top}b=0$, then $-Lx + b = \vect{0}_n$ has solutions $x =
  L^{\dagger}b+\beta \vect{1}_n$, for $\beta \in \real$. Thus, by
  Theorem~\ref{thm:1-inf}, every trajectory of~\eqref{eq:ave} converges to
  an equilibrium point in $\mathrm{span}(\vect{1}_n)$. Let $t\mapsto x(t)$
  be a trajectory of~\eqref{eq:ave}. Since $v^{\top}L = 0$, we have
  $v^{\top}x(t) = v^{\top}x_0$. Thus, $t\mapsto x(t)$ converges to the
  equilibrium point $L^{\dagger}b+(v^{\top}x_0)\vect{1}_n$.  Next, we
  define $z = x + L^{\dagger}b$ to get
  \begin{align}\label{eq:ave-trans}
    \dot{z} = - L z. 
  \end{align}
  Note $\mathrm{span}(\vect{1}_n)$ is invariant for~\eqref{eq:ave-trans}
  and consists of only equilibrium points. Moreover, for $f(z):=-Lz$ we
  have $\mu_{\infty}(\jac f(z)) =\mu_{\infty}(-L)= 0$. Thus the
  system~\eqref{eq:ave-trans} is weakly contracting with respect to
  $\ell_\infty$-norm.  Lemma~\ref{ex:useful} now implies
  \begin{align*}
    \inf & \setdef{\mu_{\infty,R}(-L)}{R\in \complex^{(n-1)\times n}, \Ker(R)=\mathrm{span}(\vect{1}_n)}
    \\ & =\alpha_{\vect{1}_n^{\perp}}(-L^{\top}) = \alpha_{\mathrm{ess}}(-L) < 0. 
  \end{align*}
  This implies that there exists $R\in \complex^{(n-1)\times n}$ such that
  $\Ker R = \mathrm{span}(\vect{1}_n)$ and $\mu_{\infty,R}(-L) \le -c$ for
  some $c>0$. Thus, the assumptions of Theorem~\ref{thm:mixed} hold for the
  norm $\|\cdot\|_{\infty}$ and the semi-norm
  $\verti{\cdot}_{\infty,R}$. In turn, every trajectory of
  system~\eqref{eq:ave-trans} converges to $\mathrm{span}(\vect{1}_n)$ with
  convergence rate $-\alpha_{\mathrm{ess}}(-L)>0$.
\end{proof}

Next, we state an analogous theorem for affine flow systems, whose proof is
omitted in the interest of brevity.

\begin{theorem}[Affine flow system]\label{thm:affine-flow}
  Under the same assumptions on $L$, $v$, and $b$ as in
  Theorem~\ref{thm:affine-avg}, the \emph{affine flow system}
  \begin{align}\label{eq:flow}
    \dot{x} = - L^{\top} x + b 
  \end{align}
  \begin{enumerate}
  \item\label{p0:contractiveflow} is weakly contracting with respect to
    $\|\cdot\|_{1}$ and semi-contracting with respect to
    $\|\cdot\|_{1,R_\epsilon}$, where $R_\epsilon$ is defined in
    Remark~\ref{remark:Reps} as function of a sufficiently small
    $\epsilon$,
      
  \item\label{p1:unboundedflow} if $\vect{1}_n^{\top}b\neq0$, then every
    trajectory is unbounded,
    
  \item\label{p2:boundedflow} if $\vect{1}_n^{\top}b=0$, then the
    trajectory starting from $x(0)=x_0$ converges to the equilibrium point
    $(L^{\top})^{\dagger}b + (\vect{1}_n^{\top}x_0)v$ with exponential
    rate $-\alpha_{\mathrm{ess}}(-L)$.
    
   \end{enumerate}
 \end{theorem}

{\color{black}In sum, Theorems~\ref{thm:affine-avg}
  and~\ref{thm:affine-flow} demonstrate how weak and semi-contraction
  theory is sufficiently powerful to fully characterize the behavior of
  affine averaging and flow systems.}
 
\subsection{Distributed primal-dual dynamics is weakly contracting}
In the second application, we study a well-known distributed implementation
of unconstrained optimization problems. Suppose that we have $n$ agents
connected though an undirected weighted connected graph with Laplacian
$L$. We want to minimize an objective function $f:\real^k\to \real$ which
can be represented as $f(x) = \sum_{i=1}^{n}f_i(x)$.  When agent $i$ has
access to only the objective function $f_i$, then minimizing $f(x) =
\sum_{i=1}^{n}f_i(x)$ can be implemented in a distributed fashion as:
\begin{equation}\label{eq:distributed}
  \begin{split}
    \min\quad  &\sum_{i=1}^{n} f_i(x_i),  \\
    &x_i = x_j,\quad\mbox{for every edge }i,j.
  \end{split}
\end{equation}
The primal-dual algorithm associated with the distributed
optimization problem~\eqref{eq:distributed} is given by
\begin{equation}\label{eq:pd}
  \begin{split}
    &\dot{x}_i = -\nabla f_i(x_i) - \sum_{i=1}^{n} a_{ij}(\nu_i-\nu_j),\\
    &\dot{\nu}_i = \sum_{j=1}^{n}a_{ij}(x_i-x_j).
  \end{split}
\end{equation}
The following theorem characterizes the global {\color{black}exponential}
convergence of the dynamics~\eqref{eq:pd}.


  \begin{theorem}[Distributed primal-dual algorithm]\label{thm:pd}
    Consider the distributed optimization problem~\eqref{eq:distributed}
    and the primal-dual dynamics~\eqref{eq:pd} over a connected undirected
    weighted graph with Laplacian matrix $L$. {\color{black}Assume that
      $f = \sum_{i=1}^{n}f_i(x)$ has a minimum $x^*\in \real^{k}$ and
    that, for each $i\in\{1,\ldots,n\}$, $f_i$ is convex, twice differentiable, and $\nabla^2
    f_i(x^*) \succ 0$.}  Then
    \begin{enumerate}
      \item\label{dpd:wk} system~\eqref{eq:pd} is weakly contracting with
      respect to~$\norm{\cdot}_{2}$,
    \item\label{dpd:convergence} {\color{black}each trajectory
      $(x(t),\nu(t))$ of~\eqref{eq:pd} converges exponentially to
      $(\vect{1}_n\otimes x^*, \vect{1}_n\otimes \nu^*)$, where
      $\nu^*=\sum_{i=1}^{n}\nu_i(0)$ with rate
      $-\subscr{\alpha}{ess}\Big(\begin{bmatrix}-\nabla^2 h(x^*) & -L\otimes
        I_k\\ L\otimes I_k & \vect{0}\end{bmatrix}\Big)$.}
    \end{enumerate}
  \end{theorem}
  \begin{proof}
    We set $x=(x_1,\ldots,x_n)^{\top}$ and
    $\nu=(\nu_1,\ldots,\nu_n)^{\top}$, and define $h(x) = \sum_{i=1}^{n} f_i(x_i)$. Then algorithm~\eqref{eq:pd} can be
    written as
    \begin{equation}\label{eq:pd-matrix}
      \begin{split}
      &\dot{x} = -\nabla h(x) - (L\otimes I_k)\nu,\\
      &\dot{\nu} = (L\otimes I_k)x.
      \end{split}
    \end{equation}
    Regarding~\ref{dpd:wk}, let
    $(\dot{x},\dot{\nu})=F_{\mathrm{PD}}(x,\nu)$. Then $\jac
    F_{\mathrm{PD}}(x,\nu) = \begin{bmatrix}-\nabla^2 h(x) & -L\otimes
      I_k\\ L\otimes I_k & \vect{0}\end{bmatrix}$ and $\mu_2(\jac
    F_{\mathrm{PD}}(x,\nu)) = 0$. Therefore, the system~\eqref{eq:pd} is
    weakly contracting with respect to the $\ell_2$-norm. Regarding
    part~\ref{dpd:convergence}, let $t\mapsto (x(t),\nu(t))$ be a trajectory of the
    system~\eqref{eq:pd}. Note that, for every $u\in \real^k$, we have
    $(\vect{1}_{n}\otimes u)^{\top} \dot{\nu}(t) = (\vect{1}_{n}\otimes
    u)^{\top} (L\otimes I_k)x(t) = \vect{0}_{nk}$. This implies that
    $\sum_{i=1}^{n}\nu_i(t) = \sum_{i=1}^{n}\nu_i(0)$ and the subspace
    \begin{align*}
      V =\bigsetdef{(x,\nu)\in \real^{nk}\times \real^{nk} }{\sum\nolimits_{i=1}^{n}\nu_i = 0} 
    \end{align*}
    is invariant for system~\eqref{eq:pd}. For the Laplacian $L=L^\top$,
    define $R_{\mcV}$ by equation~\eqref{def:RV}. Define new coordinates
    $(\tilde{x},\tilde{v})$ on $V$ by $\tilde{x}=x$ and
    $\tilde{\nu}=(R_{\mcV}\otimes I_k)\nu$. Thus, the dynamical
    system~\eqref{eq:pd} restricted to $V$ can be written in the new
    coordinate $(\tilde{x},\tilde{v})$ as
      \begin{equation}\label{eq:pd-matrix-coordinate-change}
        \begin{split}
          &\dot{\tilde{x}} = -\nabla h(\tilde{x}) - (L R^{\top}_{\mcV}\otimes I_k)\tilde{\nu},\\
          &\dot{\tilde{\nu}} = (R_{\mcV}L\otimes I_k)\tilde{x}.
        \end{split}
      \end{equation}
      Let
      $(\dot{\tilde{x}},\dot{\tilde{\nu}}):=\tilde{F}_{\mathrm{PD}}(\tilde{x},\tilde{\nu})$
      Note that $(\tilde{x},\tilde{v})=(\vect{1}_n\otimes
      x^*,\vect{0}_{(n-1)k})$ is an equilibrium point of the dynamical
      system~\eqref{eq:pd-matrix-coordinate-change} and
      \begin{align*}
        D \tilde{F}_{\mathrm{PD}}(\tilde{x},\tilde{\nu}) = \begin{bmatrix}-\nabla^2 h(\tilde{x}) & -L R^{\top}_{\mcV}\otimes
          I_k\\ R_{\mcV} L\otimes I_k & \vect{0}\end{bmatrix}.
      \end{align*}
      Again we note that $\mu_2(D
      \tilde{F}_{\mathrm{PD}}(\tilde{x},\tilde{\nu}))= 0$, for every
      $(\tilde{x},\tilde{\nu}) \in V$, and therefore, the
      system~\eqref{eq:pd-matrix-coordinate-change} is weakly contracting
      with respect to the $\ell_2$-norm. Moreover,
      $-\nabla^2h(\vect{1}_n\otimes x^*) \prec
      0$ and $\Ker(L R^{\top}_{\mcV}\otimes I_k) = \emptyset$. Thus,
      by~\cite[Lemma~5.3]{AC-BG-JC:17}, the matrix $D
      \tilde{F}_{\mathrm{PD}}(\vect{1}_n\otimes x^*,0)$ is Hurwitz and the equilibrium point
      $(\vect{1}_n\otimes x^*,\vect{0}_{(n-1)k})$ is locally asymptotically stable for the
      dynamical system~\eqref{eq:pd-matrix-coordinate-change}.
      Theorem~\ref{thm:weak-dichotomy}\ref{thmwd:p4.5} applied to the
      weakly contracting system~\eqref{eq:pd-matrix-coordinate-change}
      implies that $(\vect{1}_n\otimes x^*,\vect{0}_{(n-1)k})$ is a
      globally exponentially stable equilibrium point of the system~\eqref{eq:pd-matrix-coordinate-change}, i.e.,
      $\lim_{t\to\infty}(\tilde{x}(t),\tilde{\nu}(t)) =
      (\vect{1}_n\otimes x^*,\vect{0}_{(n-1)k})$. Therefore, we get $\lim_{t\to\infty} x(t) =
      \vect{1}_n\otimes x^*$ and
      \begin{align*}
        \vect{0}_{nk} &= (R^{\top}_{\mcV}\otimes I_k) \lim_{t\to\infty}
        \tilde{\nu} = \lim_{t\to\infty} (R^{\top}_{\mcV}R_{\mcV}\otimes
        I_k)\nu(t) \\ &= \lim_{t\to\infty}
        \left((I_n-\tfrac{1}{n}\vect{1}_n\vect{1}_n^{\top})\otimes
        I_k\right)\nu(t) \\ & = \lim_{t\to\infty} \left(\nu(t) - \vect{1}_n\otimes \nu^*\right),
      \end{align*}
      where the last equality holds because $\tfrac{1}{n}\sum_{i=1}^{n}\nu_i(t) =
      \nu^*$. As a result, we get $\lim_{t\to\infty} \nu(t) =
      \vect{1}_n\otimes \nu^*$. {\color{black} Moreover, by
        Theorem~\ref{thm:weak-dichotomy}\ref{thmwd:p4.5}, the
        convergence rate to $(\vect{1}_n\otimes x^*,
        \vect{1}_n\otimes \nu^*)$ is $-\alpha(\begin{bmatrix}-\nabla^2 h(\tilde{x}^*) & -L R^{\top}_{\mcV}\otimes
          I_k\\ R_{\mcV} L\otimes I_k & \vect{0}\end{bmatrix})$ which
        is equal to $-\subscr{\alpha}{ess}(\begin{bmatrix}-\nabla^2 h(x^*) & -L\otimes
      I_k\\ L\otimes I_k & \vect{0}\end{bmatrix})$.}
\end{proof} {\color{black}\begin{remark}[Comparison with literature]
    In~\cite{PCV-SJ-FB:19r}, using partial contraction theory, it is
    shown that the distributed primal dual algorithm~\eqref{eq:pd}
    asymptotically converges to an equilibrium point. However, the
    exponential convergences of~\eqref{eq:pd} has been only proved
    under strong convexity and Lipschitz properties of the cost
    functions $f_i$. To the best of our knowledge,
    Theorem~\ref{thm:pd} is the first result that proves exponential
    convergence of the distributed primal dual algorithm~\eqref{eq:pd}
    and provides the rate of convergence under weak
    convexity of each $f_i$.\oprocend
  \end{remark}}

\subsection{Strongly diffusively-coupled systems are semi-contracting}
As third application, we study synchronization phenomena in networks of
identical {\color{black}diffusively-coupled systems}. Consider $n$ agents connected
through a weighted undirected graph with Laplacian matrix $L$. Suppose the
agents have identical internal dynamics described by the time-varying
vector field $f:\real_{\ge 0}\times \real^k\to \real^k$. {\color{black}Then the
network dynamics is}
\begin{align}\label{eq:diffusive}
  \dot{x}_i = f(t,x_i) - \sum_{j=1}^{n} a_{ij}(x_i-x_j), \quad
  i\in\until{n}.
\end{align}
{\color{black}As mentioned in the introduction, such diffusively-coupled dynamical
systems are widespread in many disciplines including developmental
biology, neuroscience, cellular systems, and cellular neural networks.} 
Next, we introduce a novel useful norm.  For $p\in [1,\infty]$, define the
\emph{$(2,p)$-tensor norm} $\|\cdot\|_{(2,p)}$ on
$\real^{nk}\simeq\real^n\otimes \real^k$ by:
\begin{align}\label{eq:pi-norm}
  \|u\|_{(2,p)} = \inf\Bigsetdef{\Big(\sum_{i=1}^{r}\|v^i\|^2_2\|w^i\|_p^2\Big)^{\frac{1}{2}}}{u =\sum_{i=1}^{r}v^i\otimes w^i}. 
\end{align}
The {\color{black}well-definedness} and properties of the $(2,p)$-tensor
norm are studied in Lemma~\ref{lem:crossnorm} in Appendix~\ref{app:2-p-tensornorm}. {\color{black} The $(2,p)$-tensor norm is closely related to, but
  different from, the well-known projective tensor product norm
  (see~\cite[Chapter~2]{RAR:13} for definition and properties of this
  well-known norm). The $(2,p)$-tensor norm is also different from the
  mixed global norm introduced in~\cite{GR-MDB-EDS:13} for hierarchical
  analysis of network systems. For the Laplacian $L=L^\top$, define
  $R_{\mcV}$ by equation~\eqref{def:RV}.}

\begin{theorem}[{\color{black}Networks of diffusively-coupled dynamical
    systems}]\label{thm:diffusive} Consider the network of
  diffusively-coupled identical {\color{black}dynamical
    systems}~\eqref{eq:diffusive} over a connected weighted undirected
  graph with Laplacian matrix $L$. Let $Q\in \real^{k\times k}$ be an
  invertible matrix and $p\in [1,\infty]$. Suppose there exists a positive
  $c$ such that, for every $(t,x)\in \real_{\ge 0}\times \real^k$,
  \begin{align}\label{eq:sync-cond}
    \mu_{p,Q}(\jac{f}(t,x)) &\leq \lambda_2(L) - c.
   \end{align}
   Then,
  \begin{enumerate}
  \item\label{sync1} system~\eqref{eq:diffusive} is semi-contracting with
    rate $c$ with respect to
    {\color{black}$\verti{\cdot}_{(2,p),R_{\mcV}\otimes Q}$},
    \vspace{2.0pt}
    
    
  \item\label{sync2} system~\eqref{eq:diffusive} achieves global
    exponential synchronization with rate $c$, i.e., for each trajectory
    $(x_1(t),\dots,x_n(t))$ and for each pair $i,j$, the distance
    $\|x_i(t)-x_j(t)\|_2$ vanishes exponentially fast with rate $c$.
        
  \end{enumerate}
\end{theorem}
\begin{proof}{\color{black}
  Set $x = (x_1^{\top},\ldots,x_n^{\top})^{\top}\in \real^{nk}$ and
  define $\subscr{x}{ave} :=\frac{1}{n}\sum_{i=1}^{n}x_i$. The dynamics~\eqref{eq:diffusive} can be written
  as:
  \begin{align}\label{eq:diffusive-matrix}
    \dot{x} = F(t,x) - (L\otimes I_k)x,
  \end{align}
  where $F(t,x) = [f^{\top}(t,x_1),\ldots,
  f^{\top}(t,x_n)]^{\top}$. Moreover,
  \begin{align*}
    \jac F(t,x) = \begin{bmatrix} \jac f(t,x_1) & 
      \ldots & \vect{0}_{k\times k}\\ \vdots & \ddots &
      \vdots \\ \vect{0}_{k\times k} &
      \ldots & \jac f(t,x_n) 
    \end{bmatrix}\in \real^{nk\times nk}.
  \end{align*}
   Next, we show that $\mu_{(2,p), R_{\mcV}\otimes Q} (\jac F(t,x) -
   L\otimes I_k) \le -c$. First note that
   \begin{multline*}
     \mu_{(2,p),R_{\mcV}\otimes Q} (\jac F(t,x) - L\otimes I_k) \\ \le \mu_{(2,p),R_{\mcV}\otimes Q} (\jac F(t,x)) + \mu_{(2,p),R_{\mcV}\otimes Q} (- L\otimes I_k). 
   \end{multline*}
   From Lemma~\ref{ex:useful}, recall that
   $R_{\mcV}^{\dagger}=R_{\mcV}^{\top}(R_{\mcV}R_{\mcV}^{\top})^{-1} =
   R_{\mcV}^{\top}$ and
   $R_{\mcV}LR_{\mcV}^{\top}=\Lambda=\diag(\lambda_2(L),\ldots,\lambda_n(L))$. Using
  Theorem~\ref{lem:computational}\ref{p2:seminorm}, this implies that
   \begin{align}\label{eq:l-bound}
     \mu_{(2,p),R_{\mcV}\otimes Q} (- L\otimes I_k) &= \mu_{(2,p)}
     (-R_{\mcV}LR_{\mcV}^{\top}\otimes QQ^{-1}) \nonumber\\ &= \mu_{(2,p)}
     (-\Lambda\otimes I_k) \le -\lambda_2(L).
   \end{align}
   where the last inequality holds by Lemma~\ref{lem:crossnorm}\ref{p5}.  
    Moreover, we have $R_{\mcV}R_{\mcV}^{\top}= I_{n-1}$ and
    $R_{\mcV}\otimes Q = (R_{\mcV}\otimes I_k)(I_n\otimes Q)$. Thus,
    Theorem~\ref{lem:computational}\ref{p2:seminorm} and
    Lemma~\ref{lem:crossnorm}\ref{p4} imply that
    \begin{align*}
   \mu_{(2,p),R_{\mcV}\otimes Q}& (\jac F(t,x)) \\ & = \mu_{(2,p),
                                                     R_{\mcV}\otimes
                                                     I_k}((I_n\otimes
                                                     Q) \jac F(t,x)
                                                     (I_n\otimes
                                                     Q^{-1})) \\ & \le \mu_{(2,p)}((I_n\otimes
                                                     Q) \jac F(t,x)
                                                     (I_n\otimes
                                                     Q^{-1}))
    \end{align*}
    Note that
     \begin{align*}
       \Gamma:=(I_n\otimes
       Q) \jac F(t,x) (I_n\otimes Q^{-1}) = \begin{bmatrix}\Gamma_{1} &
         \ldots & \vect{0}_{k\times k}\\
         \vdots&
         \ddots & \vdots\\
         \vect{0}_{k\times k} &
         \ldots & \Gamma_{n} \end{bmatrix},
     \end{align*}
     where $\Gamma_{i}=Q\jac f(t,x_i)Q^{-1}\in \real^{k\times k}$, for every $i\in
     \{1,\ldots,n\}$. In turn, using Lemma~\ref{lem:crossnorm}\ref{p5}, 
     \begin{multline}\label{eq:f-bound}
       \mu_{(2,p), R_{\mcV}\otimes Q}(\jac F(t,x)) \le 
       \mu_{(2,p)}(\Gamma) \\\le \max_{i\in \{1,\ldots,n\}}\left\{\mu_{p}(\Gamma_i)\right\} = \max_{i\in\{1,\ldots,n\}}\mu_{p,Q}(\jac f(t,x_i)).
       \end{multline}
     Thus, combining~\eqref{eq:l-bound} and~\eqref{eq:f-bound}, we get 
  \begin{multline*}
  \mu_{(2,p), R_{\mcV}\otimes Q} (\jac F(t,x) - L\otimes I_k) \\ \le
    \max_{i\in\{1,\ldots,n\}}\mu_{p,Q}(\jac f(t,x_i)) - \lambda_2(L).
  \end{multline*}
  Using~\eqref{eq:sync-cond}, for every $t\in \real_{\ge 0}$ and every
  $x\in \real^{nk}$,
\begin{align*}
  \mu_{(2,p), R_{\mcV}\otimes Q} (\jac F(t,x) - L\otimes I_k) \le -c,
\end{align*}
This implies that the system~\eqref{eq:diffusive} is semi-contracting with respect
to $\verti{\cdot}_{(2,p), R_{\mcV}\otimes Q}$ with rate $c$. It is easy to see that
  \begin{align*}
    \Ker\verti{\cdot}_{(2,p),R_{\mcV}\otimes Q} = \Ker
  (R_{\mcV}\otimes Q) =  \mcS
  \end{align*}
  where $\mcS= \mathrm{span}\{\vect{1}_n\otimes u\mid u\in \real^k\}$
  is shifted-invariant under the dynamical system~\eqref{eq:diffusive-matrix}. Thus, 
using Theorem~\ref{thm:partialcontraction}\ref{p2:convergence}, every
trajectory of~\eqref{eq:diffusive-matrix} converges exponentially with
rate $c$ to $\mcS$. Now, define $P =
  I_{n} - \tfrac{1}{n}\vect{1}_n\vect{1}_n^{\top}$ and note that $P=P^\top$
  is the orthogonal projection onto $\vect{1}^{\perp}_n$ and $\Ker P
  = \mathrm{span}(\vect{1}_n)$. Note that $P\otimes I_k$ is the orthogonal
projection onto $\mcS^{\perp}$. Thus,
Theorem~\ref{thm:partialcontraction}\ref{p2:convergence} implies that
$\lim_{t\to\infty}\|(P\otimes I_k)x(t)\|_2 = 0$ with rate
$c$. Moreover, we know that 
\begin{align*}
  (P\otimes I_k)x(t) & = x(t) - \tfrac{1}{n}(\vect{1}_n\vect{1}_n^{\top}\otimes
I_k) x(t) \\ & = x(t) - (\vect{1}_n\otimes \subscr{x}{ave}(t)).
\end{align*}
As a result, $\lim_{t\to\infty}\|x(t) - (\vect{1}_n\otimes \subscr{x}{ave}(t))\|_2=0$
with rate $c$, which in turn means that
$\lim_{t\to\infty}\|x_i(t)-x_j(t)\|_2=0$ with rate $c$, for every $i,j\in \{1,\ldots,n\}$.}
\end{proof}
\begin{remark}[Comparison with literature]
  {\color{black} In the literature, most of the synchronization
    conditions for diffusively-coupled systems are either based on
    quadratic Lyapunov functions or $\ell_2$-norms, or are not
    applicable to networks with general topologies. It has been shown
    that in many important applications, these synchronization conditions ignore the
    structure of internal dynamics and provide conservative estimates~\cite{GR-MDB-EDS:10a,MC:07a,ZA-EDS:14}.}  For networks of
  diffusively-coupled identical {\color{black}dynamical systems}
  condition~\eqref{eq:sync-cond} has two unique features: (i) it is
  based on matrix measures induced by weighted $\ell_p$-norm, for
  $p\in [1,\infty]$, thus extending the QUAD-based global
  synchronization conditions in~\cite{WL-TC:06,PD-MdB-GR:11} and the
  $\ell_2$-norm-based synchronization condition in~\cite{MA:11}; and
  (ii) it is applicable to networks with arbitrary undirected
  topology, thus generalizing the synchronization conditions developed
  for specific network topologies in~\cite{ZA-EDS:14}. \oprocend
\end{remark}

\begin{remark}[Diffusively-coupled linear systems]
  Given $\mathcal{A}\in \real^{k\times k}$ and a weighted undirected graph,
  consider the diffusively-coupled linear identical systems:
  \begin{align}\label{eq:diff-linear}
    \dot{x}_i = \mathcal{A}x_i -\sum_{j=1}^{n}a_{ij}(x_i-x_j),\quad i\in
    \until{n}.
   \end{align}
  It is known that the system~\eqref{eq:diff-linear} achieves
  synchronization if and only if $\mathcal{A}-\lambda_2(L)I_k$ is Hurwitz,
  e.g., see~\cite[Theorem 8.4(ii)]{FB:20}. For a fixed $p\in [1,\infty]$,
  condition~\eqref{eq:sync-cond} reads
    \begin{align*}
     \mu_{p,Q}(\mathcal{A})<\lambda_2(L).
    \end{align*}
    Using Theorem~\ref{thm:abscissa}\ref{omm:p2:R}, we get that $\inf_{Q
      \text{ invertible}} \mu_{p,Q}(\mathcal{A}) =
    \alpha(\mathcal{A})$. Thus, per Theorem~\ref{thm:diffusive}\ref{sync2},
    if $\alpha(\mathcal{A})< \lambda_2(L)$, then the diffusively-coupled
    linear identical systems~\eqref{eq:diff-linear} achieve synchronization
    with exponential rate $\lambda_2(L)-\mu_{p,Q}(\mathcal{A})$. In other
    words, condition~\eqref{eq:sync-cond} recovers the exact threshold of
    synchronization for diffusively-coupled linear systems. We note that
    the contraction-based synchronization conditions for
    {\color{black}diffusively-coupled identical dynamical systems} proposed
    in~\cite{MA:11,PD-MdB-GR:11,ZA-EDS:14} do not recover this threshold
    for the linear systems. \oprocend
  \end{remark}

\section{Conclusion}
In this paper we have provided multiple analytic extensions of the basic
ideas in contraction theory, including advanced results on
semi-contracting, weakly contracting and doubly-contracting systems. We
have also illustrated how to apply our results to various network systems.
Possible directions for future work include extensions to discrete-time and
hybrid systems, to differential-geometric treatments, and to
infinite-dimensional systems.

\section{Acknowledgments}  {\color{black}
The authors thank Zahra Aminzare for suggestions about the proof of
Theorem~\ref{thm:diffusive} and Samuel Coogan for suggestions about
Theorem~\ref{thm:partialcontraction}.  The third author also thanks Mario
Di~Bernardo, Giovanni Russo, Rodolphe Sepulchre, and John W.\ Simpson-Porco
for discussions about contraction theory. }

\appendices

\section{Semi-contraction vs. horizontal contraction}\label{app:semi-horizontal}
{\color{black} In this appendix, we study the connection between the
  semi-contraction setting in Section~\ref{sec:partial} and the horizontal
  contraction framework developed in~\cite{FF-RS:14}. We consider the
  Euclidean space $\real^n$ equipped with a semi-norm $\verti{\cdot}$ and
  define a Finsler structure on $\real^n$ compatible with the semi-norm
  $\verti{\cdot}$. Specifically, assume $\mathcal{H}$ is a vector
  subspace of $\real^n$ such that
  \begin{align}\label{eq:kerneldecom}
    \real^n = \Ker\verti{\cdot}\oplus \mathcal{H}.
  \end{align}
  Let $\pi_{\mathcal{H}}:\real^n\to \real^n$ denote the oblique
  projection onto $\mathcal{H}$ parallel to $\Ker\verti{\cdot}$ and consider a
  norm $\|\cdot\|$ on $\real^n$ such that
  \begin{align*}
    \|\pi_{\mathcal{H}}(v)\|=\verti{v},\quad\mbox{ for all }v\in \real^n.
  \end{align*}
  It is easy to see that the manifold $\real^n$ with the norm $\|\cdot\|$
  is a Finsler manifold. Considering the decomposition~\eqref{eq:kerneldecom} as a
  vertical/horizontal decomposition of the tangent space
  $\real^n$~\cite[Definition 4]{FF-RS:14}, the associated horizontal
  projection is given by the oblique projection
  $\pi_{\mathcal{H}}$~\cite[Theorem 4]{FF-RS:14} and the associated pseudo-distance is given by $d(x_1,x_2) =
  \verti{x_1-x_2}$~\cite{FF-RS:14}. 
Thus, we can restate~\cite[Theorem 4]{FF-RS:14}, for the Finsler manifold
$(\real^n,\|\cdot\|)$ with the vertical/horizontal
decomposition~\eqref{eq:kerneldecom}.

\begin{theorem}[Horizontal contraction on Euclidean spaces {\cite[Theorem 4]{FF-RS:14}}]\label{cor}
  Consider the dynamical system~\eqref{eq:TV-nonlinear} on the Finsler
  manifold $(\real^n,\|\cdot\|)$ and assume that there exists $c>0$ such
  that the following conditions hold, for every $(t,x)\in \real_{\ge
    0}\times \real^n$:
  \begin{align}
    \mu_{\verti{\cdot}}(\jac f(t,x)) &\le -c, \\
    \pi_{\mathcal{H}} \jac f(t,x) &= Df(t,x) \pi_{\mathcal{H}},\label{eq:invariance-strong}
  \end{align}
Then, there exists $K\ge 1$ such that, for every $x_0,y_0\in \real^n$,
  \begin{align*}
    \verti{\phi(t,x_0)-\phi(t,y_0)} \le Ke^{-ct} \verti{x_0-y_0},\quad\mbox{for all } t\ge 0. 
\end{align*}
\end{theorem}
Now, we compare Theorem~\ref{thm:partialcontraction}\ref{p1:estimates}
and Theorem~\ref{cor}. 
\begin{remark}[Comparison]
  \begin{enumerate}
  \item The commutativity condition~\eqref{eq:invariance-strong} in
    Theorem~\ref{cor} depends on the choice of the horizontal subspace
    $\mathcal{H}$ in the decomposition~\eqref{eq:kerneldecom}. On the other
    hand, the infinitesimally invariant condition~\eqref{InfI} in
    Theorem~\ref{thm:partialcontraction}\ref{p1:estimates} only depends on
    $\Ker\verti{\cdot}$ and is independent of this
    decomposition. Next, we show that condition~\eqref{eq:invariance-strong} is a stronger
    requirement than the infinitesimal invariance
    condition~\eqref{InfI}. Assume that
    condition~\eqref{eq:invariance-strong} holds. Then, for every $v\in \Ker\verti{\cdot}$ and every $(t,x)\in \real_{\ge 0}\times \real^n$,
\begin{align*}
      \pi_{\mathcal{H}} \jac f(t,x) v = \jac f(t,x) \pi_{\mathcal{H}} v =  Df(t,x)
  \vect{0}_n = \vect{0}_n. 
\end{align*}
This implies that $\jac f(t,x)v\in \Ker\verti{\cdot}$, for every $(t,x)\in
\real_{\ge 0}\times \real^n$. As a result, if
condition~\eqref{eq:invariance-strong} holds then 
$\Ker\verti{\cdot}$ is
infinitesimally invariant under the
dynamics~\eqref{eq:TV-nonlinear}. Example~\ref{eq:example-P1-horizontal}
below shows that, in general, the converse is not true and infinitesimal
invariance of $\Ker\verti{\cdot}$ does not imply the commutativity
condition~\eqref{eq:invariance-strong}.
  
\item Theorem~\ref{cor} provides a weaker result
  compared to
  Theorem~\ref{thm:partialcontraction}\ref{p1:estimates}. In fact,
  Theorem~\ref{cor} guarantees that there exists $K\ge 1$ such that,
  for every $x_0,y_0\in \real^n$,
  \begin{align*}
    \verti{\phi(t,x_0)-\phi(t,y_0)} \le Ke^{-ct} \verti{x_0-y_0},\;\mbox{
    for all } t\ge 0. 
  \end{align*}
  On the other hand, Theorem~\ref{thm:partialcontraction}\ref{p1:estimates} ensures
  that every two trajectories satisfy condition~\eqref{eq:semi-contraction-traj}. 
    \end{enumerate}
  \end{remark}

  \begin{example}\label{eq:example-P1-horizontal}
    Consider the following linear system on $\real^2$:
    \begin{align}
      \begin{split}\label{eq:linear2x2}
      \dot{x}_1 &= -x_1 := f_1(x_1,x_2),\\
      \dot{x}_2 &= x_1 - 2x_2:=f_2(x_1,x_2)
      \end{split}
    \end{align}
    with the semi-norm defined by $\verti{(x_1,x_2)^{\top}} = |x_1|$, for
    every $(x_1,x_2)^{\top}\in \real^2$. Note that
    $\Ker\verti{\cdot}=\mathrm{span}\{(0,1)^{\top}\}$. We consider the
    vertical/horizontal decomposition~\eqref{eq:kerneldecom} with
    $\mathcal{H}=\Ker\verti{\cdot}^{\perp} =
    \mathrm{span}\{(1,0)^{\top}\}$. Then, the horizontal projection is
    given by $\pi_{\Ker\verti{\cdot}^{\perp}} = \begin{bmatrix}
      1 & 0 \\
      0 & 0
    \end{bmatrix}$. Note that
    $\jac f(x_1,x_2)\begin{bmatrix}0 &1\end{bmatrix}^{\top} = \begin{bmatrix}0 &-2\end{bmatrix}^{\top}$ and therefore
    $\Ker\verti{\cdot}$ is infinitesimally invariant under the
    system~\eqref{eq:linear2x2}. However, the system~\eqref{eq:linear2x2} does not satisfy
    condition~\eqref{eq:invariance-strong}, since $
          \pi_{\Ker\verti{\cdot}^{\perp}}\jac f(x_1,x_2) = \begin{bmatrix}
            -1 & 0 \\
            0 & 0
          \end{bmatrix}  \ne \begin{bmatrix}
            -1 & 0 \\
            1 & 0
          \end{bmatrix} = \jac
          f(x_1,x_2)\pi_{\Ker\verti{\cdot}^{\perp}}$.
    \end{example}}

\section{A Useful lemma}\label{app:2-p-tensornorm}

\begin{lemma}[Properties of $(2,p)$-tensor norms]\label{lem:crossnorm}
  Consider the identification $\real^{nk}\simeq \real^n\otimes
  \real^k$ and the $(2,p)$-tensor norm defined in~\eqref{eq:pi-norm}.  Let $A\in \real^{nk\times nk}$, $\Lambda =
  \mathrm{blkdg}(\Lambda_1,\ldots,\Lambda_n)$ such that $\Lambda_i\in
  \real^{k\times k}$, for every $i\in \{1,\ldots,n\}$ and 
  $U\in \real^{(n-1)\times n}$ is such that $UU^{\top}=I_{n-1}$, then
  \begin{enumerate}
  \item\label{p1} $(2,p)$-tensor norm is well-defined;
    \item \label{p2} $\|U\otimes I_k\|_{(2,p)} = \|U^{\top}\otimes I_k\|_{(2,p)}
      =1$; 
  \item \label{p3} $\|\Lambda\|_{(2,p)} \le \max_{i}\|\Lambda_i\|_p$;
  \item \label{p4} $\mu_{(2,p), U\otimes I_k}(A) \le \mu_{(2,p)}(A)$; 
    \item \label{p5} $\mu_{(2,p)}(\Lambda)\le \max_{i}\mu_p(\Lambda_i)$. 
  \end{enumerate}
\end{lemma}
\begin{proof}
  Regarding part~\ref{p1}, we follow the argument
  in~\cite[Proposition 2.1]{RAR:13}
  for projective tensor product norm to show that $(2,p)$-tensor norm
  is a norm. Using the exact same argument as
  in~\cite[Proposition 2.1]{RAR:13}, it is straightforward to see
  that, for every $c\in \real$, we have
  $\|cu\|_{(2,p)}=|c|\|u\|_{(2,p)}$. Now we show that $\|\cdot\|_{(2,p)}$ satisfies triangle
inequality. Let  $u,z\in \real^{n}\otimes \real^k$ and let
$\epsilon>0$. Then by definition~\eqref{eq:pi-norm}, there exist
representations $u = \sum_{i=1}^{r} v^i\otimes
  w^i$ and $z=\sum_{i=1}^{s} x^i\otimes
  y^i$ such that 
\begin{align*}
\big(\sum_{i=1}^{r}
  \|v^i\|^2_2\|w^i\|^2_p\big)^{\frac{1}{2}}&\le
  \|u\|_{(2,p)}+\frac{\epsilon}{2},\\
\big(\sum_{i=1}^{s}
  \|x^i\|^2_2\|y^i\|^2_p\big)^{\frac{1}{2}}&\le
                                                      \|z\|_{(2,p)}+\frac{\epsilon}{2}. 
\end{align*}
Note that $\sum_{i=1}^{s}x^i\otimes y^i + \sum_{i=1}^{r} v^i\otimes w^i$ is a representation of $u+z$. As a result, we have 
\begin{align*}
\|u+z\|_{(2,p)} &\le \big(\sum_{i=1}^{s}\|x^i\|^2_2\|y^i\|_p^2 + \sum_{i=1}^{r}\|v^i\|_2^2\|w^i\|_p^2\big)^{\frac{1}{2}}\\ &
                                                                       \le \big(\sum_{i=1}^{s}\|x^i\|^2_2\|y^i\|_p^2\big)^{\frac{1}{2}}
  +\big(\sum_{i=1}^{r}\|v^i\|_2^2\|w^i\|_p^2\big)^{\frac{1}{2}}\\ & \le   \|z\|_{(2,p)}+\|u\|_{(2,p)}+\epsilon
\end{align*}
Since $\epsilon$ can be chosen arbitrarily small, then we have
$\|u+z\|_{(2,p)}\le \|z\|_{(2,p)}+\|u\|_{(2,p)}$. Finally, we show that if $\|u\|_{(2,p)}=0$, then we have $u=\vect{0}_{nk}$. Since
$\|u\|_{(2,p)}=0$, there exists a representation $u = \sum_{i=1}^{r} v^i\otimes
  w^i$ such that $\sum_{i=1}^{r} \|v^i\|^2_2\|w^i\|_p^2 \le \epsilon$. Let
$\phi:\real^n\to\real$ and $\psi:\real^k\to \real$ be two linear
functionals. Then we have $\sum_{i=1}^{r}(\phi(v^i))^2(\psi(w^i))^2
\le \epsilon \|\phi\|^2_2\|\psi\|^2_p$. Since this inequality holds
for every $\epsilon>0$, then we have
$\sum_{i=1}^{r}(\phi(v^i))^2(\psi(w^i))^2 = 0$. This implies that, for
every $i\in \{1,\ldots,r\}$, either $\phi(v^i)=0$ or
$\psi(w^i)=0$. Thus, $\sum_{i=1}^{r} \phi(v^i)\psi(w^i) =
0$. Now using~\cite[Proposition 1.2]{RAR:13}, we get $u=\vect{0}_{nk}$. Regarding part~\ref{p2}, note that, for $u\in \real^{n}\otimes \real^k$ with a
  representation $u=\sum_{i=1}^{r}v^i\otimes w^i$, 
  \begin{align*}
    \|(U\otimes I_k)&u\|^2_{(2,p)} = \big\|\sum_{i=1}^{r} (Uv^i)\otimes
    w^i\big\|^2_{(2,p)} \\ & \le \sum_{i=1}^{r} \|Uv^i\|^2_2\|w^i\|^2_p \le 
    \sum_{i=1}^{r} \|v^i\|^2_2\|w^i\|^2_p,
  \end{align*}
  where the first inequality holds by definition of $(2,p)$-tensor
  norm and the second inequality holds because, for every
  $i\in \{1,\ldots,r\}$, we have $\|Uv^i\|_2\le \|U\|_2\|v^i\|_2$ and
  $\|U\|_2=\max_{i\in \{1,\ldots,n-1\}}\lambda_i(UU^{\top})=1$. Since the above inequality
  holds for every representation of $u$, then
  $\|(U\otimes I_k)u\|_{(2,p)}\le \|u\|_{(2,p)}$, for every
  $u\in \real^n\otimes \real^k$. This implies that
  $\|U\otimes I_k\|_{(2,p)}\le 1$. Similarly, one can show that
  $\|U^{\top}\otimes I_k\|_{(2,p)}\le 1$. However,
  $(U\otimes I_k)(U^{\top}\otimes I_k)=I_{nk}$, this implies that
  $\|U\otimes I_k\|_{(2,p)}\|U^{\top}\otimes I_k\|_{(2,p)}\ge 1$ and
  in turn
  $\|U\otimes I_k\|_{(2,p)} =\|U^{\top}\otimes I_k\|_{(2,p)} = 1$. Regarding part~\ref{p3}, for $u\in \real^{n}\otimes \real^k$ with a
  representation $u=\sum_{i=1}^{r}v^i\otimes w^i$, we have $\Lambda u
  = \sum_{i=1}^{r} \sum_{j=1}^{n}v^i_je_j\otimes \Lambda_{j}w^i$, where $e_k$ is the
    $k$th standard basis in $\real^n$, for every $k\in
    \{1,\ldots,n\}$. This implies that 
  \begin{align*}
    \|\Lambda u\|&^2_{(2,p)} = \big\|\sum_{i=1}^{r}\sum_{j=1}^{n}v^i_je_j\otimes \Lambda_{j}w^i\big\|^2_{(2,p)} 
                             \\ &\le \sum_{i=1}^{r}\sum_{j=1}^{n}|v^i_j|^2_2\|\Lambda_jw^i\|^2_p
                             \le
                             \sum_{i=1}^{r}\sum_{j=1}^n|v^i_j|^2_2\|\Lambda_j\|^2_p\|w^i\|^2_p\\
                             &\le
                             \max_{i}\|\Lambda_i\|^2_p\sum_{i=1}^{r}\sum_{j=1}^{n}
                             |v^i_j|^2_2\|w^i\|^2_p\\ &=\max_{i}\|\Lambda_i\|^2_p\sum_{i=1}^{r}\|v^i\|^2_2\|w^i\|^2_p,
  \end{align*}
  where the first inequality holds by the definition of $(2,p)$-tensor
  norm. For the second inequality, we used
  $\|\Lambda_jw^i\|_p\le \|\Lambda_j\|_p\|w^i\|_p$, for every
  $i\in \{1,\ldots,r\}$ and every $j\in \{1,\ldots,n\}$. For the third
  inequality, we used $\|\Lambda_j\|_p\le
  \max_i\|\Lambda_i\|_p$. Finally, for the last equality, we used the
  fact that $\sum_{j=1}^{n} |v^i_j|^2 = \|v^i\|_2^2$.  Since the above
  inequality holds for every representation of $u$, then we have
  $\|\Lambda u\|_{(2,p)} \le \max_{i}\|\Lambda_i\|_p\|u\|_{(2,p)}$,
  for every $u\in \real^n\otimes \real^k$. This implies that
  $\|\Lambda\|_{(2,p)}\le \max_{i}\|\Lambda_i\|_p$. Regarding
  part~\ref{p4}, note that $U^{\dagger}=U^{\top}$. As a result, using
  Theorem~\ref{lem:computational}\ref{p2:seminorm}, we have
  \begin{multline*}
    \mu_{(2,p),U\otimes I_k} (A) = \mu_{(2,p)}((U\otimes I_k)A(U^{\top}\otimes I_k)) \\ =
    \lim_{h\to 0^{+}} \frac{\|I_{(n-1)k}+h(U\otimes I)A(U^{\top}\otimes I_k)\|_{(2,p)}-1}{h}.
  \end{multline*}
  Note that, for every $h\in \real_{\ge 0}$, 
\begin{align*}
  \|&I_{(n-1)k}-h(U\otimes I)A(U^{\top}\otimes
  I_k)\|_{(2,p)} \\ & = \|(U\otimes I_k)(U^{\top}\otimes I_k)+h(U\otimes
  I_k)A(U^{\top}\otimes I_k)\|_{(2,p)} \\ &\le \|(U\otimes
  I_k)\|_{(2,p)}\|I_{nk}+ h A\|_{(2,p)}\|(U^{\top}\otimes
  I_k)\|_{(2,p)} \\ &= \|I_{nk}+ h A\|_{(2,p)},
\end{align*}
where for the first equality, we used the fact that $UU^{\top}=I_n$,
and the last equality holds by part~\ref{p2}. As a result, we get
$\mu_{(2,p),U\otimes I_k} (A)  \le \mu_{(2,p)} (A)$. Regrading part~\ref{p5}, note that
\begin{align*}
  \mu_{(2,p)}(\Lambda) = \lim_{h\to 0^{+}}
  \frac{\|I_{nk}+h\Lambda\|_{(2,p)}-1}{h}. 
  \end{align*}
On the other hand, we have 
$I_{nk}+h\Lambda=\mathrm{blkdg}(I_k+h\Lambda_1,\ldots,
I_{k}+h\Lambda_{n})$. Thus, by part~\ref{p3}, we have 
\begin{align*}
  \mu_{(2,p)}(\Lambda) & = \lim_{h\to 0^{+}}
  \frac{\|I_{nk}+h\Lambda\|_{(2,p)}-1}{h} \\ & \le \lim_{h\to 0^{+}}
  \frac{\max_{i}\|I_{k}+h\Lambda_i\|_{p}-1}{h} \\ & = \max_{i}\left\{\lim_{h\to 0^{+}}
  \frac{\|I_{k}+h\Lambda_i\|_{p}-1}{h}\right\}  \\ & =
  \max_{i}\mu_{p}(\Lambda_i). \qedhere
  \end{align*}
\end{proof}

\section{Proof of Lemma~\ref{lem:weak-implies-invariant}}

\begin{proof}[Proof of Lemma~\ref{lem:weak-implies-invariant}]
  Since the trajectory $t\mapsto x(t)$ is bounded in $\real^n$, there
  exists $r>0$ such that $t\mapsto x(t)$ is contained in
  $\overline{B}_{\|\cdot\|}(x(0),r)\cap C$. For every $s\ge 0$, we define
  the set $\mcU_s = \bigcap_{t\ge s
  }\left(\overline{B}_{\|\cdot\|}(x(t),r)\cap C\right)$. It is easy to show
  the family of sets $\{\mcU_s\}_{s\ge 0}$ satisfies the following
  monotonicity property: $\mcU_{s_1}\subseteq\mcU_{s_2}$, for
  $s_1<s_2$. Additionally, we define the sets $U$ and $V$ by
  \begin{align*}
    U = \bigcup\nolimits_{s\ge 0 }\mcU_s\qquad V= \bigcup\nolimits_{t\ge 0 }\left(\overline{B}_{\|\cdot\|}(x(t),r)\cap C\right).
  \end{align*}
  Since $t\mapsto x(t)$ is bounded, the set $V$ is bounded and therefore
  $\mathrm{cl}(V)$ is a compact in $\real^n$. Moreover, it is easy to
  check that $U\subseteq \mathrm{cl}(V)$. We set $W = \mathrm{cl}(U)$ and show
  that $W$ is a non-empty compact convex bounded set with the
  property that $\phi(t,W)\subseteq W$, for every $t\ge 0$.

 \emph{(Step 1: $W$ is non-empty)}. Consider $ 0=t_0\le
 t_1<\ldots<t_n$ and define the set $S= \bigcap_{i=0}^{k}\left(\overline{B}_{\|\cdot\|}(x(t_i),r)\cap C\right)$
  Let $t\ge t_k$. For every $i\in \{1,\ldots,k\}$, we have
  \begin{align}\label{eq:good}
    \|x(t) - x(t_i)\|\le \|x(t-t_i)-x(0)\|\le r,
  \end{align}
  where the first inequality holds because the system~\eqref{eq:nonlinear}
  is weakly contracting and the last inequality holds because
  $\overline{B}_{\|\cdot\|}(x(0),r)$ contains the trajectory $t\mapsto
  x(t)$. Therefore, the inequalities in~\eqref{eq:good} implies that $x(t)
  \in S$, for every $t\ge 0$. Thus $\mcS$ is non-empty. Now consider the family of sets
  $\left\{\left(\overline{B}_{\|\cdot\|}(x(t),r)\cap C\right)\right\}_{t\ge 0}$. By the
  above argument, this family is inside the compact set
  $\mathrm{cl}(V)$ and, for every finite set $J\in \real_{\ge 0}$, the intersection $\bigcap_{t_i\in J}\left(\overline{B}_{\|\cdot\|}(x(t),r)\cap C\right)$ is non-empty. Therefore, by~\cite[Theorem 26.9]{JM:00}, for every $s\ge
  0$, the set $\mcU_s$ is non-empty and, as a result, $W$ is non-empty.

  \emph{(Step 2: $W$ is convex and compact)}. We start by showing that $U$ is
  convex. Note that intersection of any collection of convex sets is convex~\cite{SB-LV:04}. Therefore, $\mcU_s$ is
convex, for every $s\ge 0$. Pick $x_1,x_2\in U$ and $\alpha\in
[0,1]$. We show that $\alpha x_1 + (1-\alpha)x_2\in U$. By definition
of $U$, there exists $s_1,s_2\ge 0$ such that $x_1\in \mcU_{s_1}$ and $x_2\in
\mcU_{s_2}$. Using the monotonicity of the family $\{\mcU_s\}_{s\ge 0}$, we have that $x_1,x_2\in \mcU_{(s_1+s_2)}$.
Since $\mcU_{(s_1+s_2)}$ is convex, we have $\alpha x_1 + (1-\alpha)x_2\in \mcU_{(s_1+s_2)}\subseteq U$. This
means that $U$ is convex. Moreover, the closure of any convex set is again convex,
thus $W = \mathrm{cl}(U)$ is convex. Since
$U\subseteq V$ and $V$ is compact, we have $W = \mathrm{cl}(U)\subseteq
V$. Thus, $W$ is a closed subset of a compact set and it is
compact.

\emph{(Step 3: $W$ is invariant under~\eqref{eq:nonlinear})}. Let $y\in U$. By definition,
there exists $s\ge 0$ such that $y\in \mcU_s$. For every $\eta>0$ and
every $t\ge s$, we have
\begin{align*}
  \|\phi(\eta,y) - x(\eta+t)\| \le \|y - x(t)\|\le r,
\end{align*}
where the first inequality holds because the system~\eqref{eq:nonlinear} is weakly contracting and
the last inequality holds because $y\in U$. This implies that, for
every $\eta\ge 0$, we have $\phi(\eta,y) \in \mcU_{\eta+s}$. This implies that
$\phi(\eta,y) \in U$. So we have $\phi(\eta,U)\subseteq U$, for every
$\eta>0$. By a simple continuity argument and using the fact that $W$
is closed, we get $\phi(\eta,W)\subseteq W$. \end{proof}

\section{Proofs of results in Section~\ref{sec:semi-measure}}\label{app:Theorem5-6-proofs}

We start this appendix with the following useful lemma.
\begin{lemma}\label{lem:useful}
Let $A\in \complex^{n\times n}$ and $R\in \complex^{k\times n}$ is
a full rank matrix such that $\Ker(R)$ is invariant under $A$. Then 
\begin{align*}
\mathrm{spec}_{\Ker R^{\perp}}(A^H) = \mathrm{spec} (RA^HR^{\dagger}). 
\end{align*}
\end{lemma}
\begin{proof}
Suppose that $\lambda\in \mathrm{spec}_{\Ker R^{\perp}}(A^H) $. This means that there exist
$v\perp\Ker(R)$ and $\lambda\in \complex$ such that $A^{H}v=\lambda v$. Since $R^{\dagger}R$ is
an orthogonal projection and $v\perp\Ker(R)$, there exists $x\in
\complex^n$ such that $v = R^{\dagger}Rx$. Thus,  $A^{H}R^{\dagger}R x = \lambda R^{\dagger}Rx$. Multiplying both side by $R$, we get
$RA^{H}R^{\dagger}Rx = \lambda RR^{\dagger}Rx = \lambda Rx$. This means that $\lambda\in
\mathrm{spec} (RA^{H}R^{\dagger})$. Thus, we have $\mathrm{spec}_{\Ker R^{\perp}}(A^H) \subseteq
\mathrm{spec} (RA^{H}R^{\dagger})$. On the other hand, since
$A\Ker(R)\subseteq \Ker(R)$ and $R$ is a full-rank matrix, we have
$|\mathrm{spec}_{\Ker R^{\perp}}(A^H)|=k$. It is easy to see that since
$RA^{H}R^{\dagger}\in \complex^{k\times k}$, we have
$\left|\mathrm{spec} (RA^{H}R^{\dagger})\right|=\left|\mathrm{spec}_{\Ker R^{\perp}}(A^H)\right|=k$. As a result $\mathrm{spec}_{\Ker R^{\perp}}(A^{H}) = \mathrm{spec} (RA^{H}R^{\dagger})$.
\end{proof}
Now, we go back to proof of Theorem~\ref{thm:measureporp}. 
\begin{proof}[Proof of Theorem~\ref{thm:measureporp}]
  Regarding part~\ref{p1:well-def}, define the function $f:\real_{\ge 0}\to
  \real$ by $f(h) := \frac{\verti{I_n +hA} - 1}{h}$. It suffices to show
  that the limit $\lim_{h\to 0^+} f(h)$ exists. First note that, for every
  $h\in \real_{\ge 0}$, we have $\verti{I_n +hA} \ge \verti{I_n} -
  \verti{hA} = 1 - h \verti{A}$. This implies that $f(h)\ge -\verti{A}$, for
  every $h\in \real_{\ge 0}$. Thus $f$ is bounded below. Moreover, for
  every $0<h_1<h_2$, we have
  \begin{multline*}
    \verti{(1/h_1)I_n +A} - \verti{(1/h_2)I_n +A} \\ \le
    \verti{(1/h_1-1/h_2)I_n } = 1/h_1-1/h_2.
  \end{multline*}
  As a result, we get
  \begin{align*}
    f(h_1) &= \verti{(1/h_1)I_n +A} - (1/h_1) \\ & \le \verti{(1/h_2)I_n
      +A} - (1/h_2) = f(h_2).
  \end{align*}
  Therefore $f:\real_{\ge 0}\to \real$ is a strictly increasing function
  which is bounded below. Thus, $\lim_{h\to 0^+} f(h)$ exists. Regarding
  part~\ref{p2:tri}, the result is straightforward using the triangle
  inequality (Theorem~\ref{induced-semi-norm}\ref{p3:triangle1}) for the
  induced semi-norm.  Regarding part~\ref{p4:lipschitz}, the proof is
  straightforward using the triangle inequality proved in
  part~\ref{p2:tri}. Finally, regarding part~\ref{p3:realpart}, suppose
  that $\verti{\cdot}_\complex$ is the complexification of
  $\verti{\cdot}$ and $\lambda\in \complex$ is an eigenvalue of $A^\top$ with the right
  eigenvector $v\in \Ker\verti{\cdot}^{\perp}\subset\complex^n$ normalized
  such that $\verti{v}_{\complex}=1$. Then we have
  \begin{align*}
    \Re(\lambda) &= \lim_{h\to 0^+}\frac{\verti{v+h\lambda v}_\complex-1}{h} =
    \lim_{h\to 0^+}\frac{\verti{v+h A^\top v}_\complex-1}{h} \\
    & \le \lim_{h\to0^+}\frac{\verti{I_n+h A^{\top}}-1}{h} = \mu_{\verti{\cdot}}(A^{\top}).\qedhere
  \end{align*}\end{proof}
Finally, we present the proof of Theorem~\ref{lem:computational}. 
\begin{proof}[Proof of Theorem~\ref{lem:computational}]
    First note that {\color{black}$\Ker\verti{\cdot}_R = \setdef{v\in \real^n}{\verti{v}_R =
    0}$.} Moreover, we have {\color{black}$\verti{v}_R = \|Rv\|$} and since
    $\|\cdot\|$ is a norm, we get {\color{black}$\verti{v}_R=0$} if and only if
    $Rv=\vect{0}_k$. This implies that {\color{black}$\Ker\verti{\cdot}_R = \Ker R$.} Regarding part~\ref{p1:seminorm}, by definition of the
    induced-norm we have {\color{black}
    \begin{align*}
      \verti{A}_R =\setdef{\|RAv\|}{\|Rv\|=1, \; v\perp\Ker R}.
    \end{align*}}
    Note that $R^{\dagger}R$ is an orthogonal projection and
    $\Ker(R^{\dagger}R) = \Ker(R)$. Thus, for every $v\perp \Ker(R)$, there exists $x\in \real^n$
    such that $v = R^{\dagger}R x$. As a result,{\color{black}
    \begin{align*}
      \verti{A}_R = \setdef{\|RAR^{\dagger}Rx\|}{\|Rx\|=1},
    \end{align*}}
    where for the above equality, we used the fact that $RR^{\dagger}R
    = R$. Since $R$ is full rank, for every $y\in \real^k$, there
    exists $x\in \real^n$ such that $y = Rx$. This means that {\color{black}
    \begin{align*}
     \verti{A}_R & = \sup\setdef{\|RAR^{\dagger}Rx\|}{\|Rx\|=1}\\
      &=\sup\setdef{\|RAR^{\dagger}y\|}{\|y\|=1} = \|RAR^{\dagger}\|                  
    \end{align*}}
    Regarding part~\ref{p2:seminorm}, note that {\color{black}
    \begin{align*}
     \mu_R(A) &= \lim_{h\to 0^+} \frac{\verti{I_n + hA}_R -1}{h}\\ & =\lim_{h\to
      0^+} \frac{\|RR^{\dagger} + hRAR^{\dagger}\| -1}{h} \\ & = \lim_{h\to
      0^+} \frac{\|I_k + hRAR^{\dagger}\| -1}{h} =
      \mu(RAR^{\dagger}),
    \end{align*}}
    where in the third equality, we used the fact that $R$ is full 
    rank and therefore $RR^{\dagger}=I_k$. Regarding
    part~\ref{p4:1-norm}, define $\widehat{\xi}\in \real_{\ge 0}^n$ by
    $\widehat{\xi}_i=\xi_i^{-1}$ if $\xi_i>0$ and
    $\widehat{\xi}_i= 0$ if $\xi_i=0$. Assume that $\xi$ has $r\le n$
    non-zero entries and define $Q_{\xi}\in
    \real^{r\times n}$ ($Q_{\widehat{\xi}}\in \real^{r\times n}$) as the matrix whose $i$th row is the $i$th
    non-zero row of $\diag(\xi)$ ($\diag(\widehat{\xi})$). Then it is easy to see that $Q_{\xi}$ is
    full rank and $Q_{\xi}^{\dagger} = Q^{\top}_{\widehat{\xi}}$. Note that {\color{black}$\verti{v}_{1,\diag(\xi)} = \|\diag(\xi)v\|_1= \|Q_{\xi}v\|_1$}, for every $v\in \real^n$. Therefore, by part~\ref{p2:seminorm},
    \begin{align*}
      \mu_{1,\diag(\xi)}(A) = \mu_1(Q_{\xi} A
      Q_{\xi}^{\dagger}) = \mu_1(Q_{\xi} A Q^{\top}_{\widehat{\xi}}). 
    \end{align*}
    The result follows using the formula for $\ell_1$-norm matrix
    measure. The proof of part \ref{p4:inf-norm} is similar
    to~\ref{p4:1-norm} and we omit it. Regarding part~\ref{p5:2-norm},
    suppose that $\mu_{2,R}(A)\le c$. Then, by part~\ref{p2:seminorm}, we
    have $\mu_{2}(RAR^{\dagger})\le c$ and, in turn,
    $\lambda_{\max}(RAR^{\dagger}+(R^{\dagger})^{\top}A^{\top}R^{\top}) \le
    2c$. Since $RAR^{\dagger}+(R^{\dagger})^{\top}A^{\top}R^{\top}$ is
    symmetric, we obtain
    \begin{align}\label{eq:inequality-LMI}
      RAR^{\dagger}+(R^{\dagger})^{\top}A^{\top}R^{\top} \preceq 2c
      I_{k}.
    \end{align}
    Multiplying the inequality~\eqref{eq:inequality-LMI} from the right by
    $R^{\top}$ and from the left by $R$,
    \begin{align}\label{eq:inequality-LMI2}
      PA(R^{\dagger}R)+(R^{\dagger}R)^{\top}A^{\top}P \preceq 2c P. 
    \end{align}
    Since $R\in\real^{k\times{n}}$ is full rank, we have
    $(\Ker(P))^{\perp}=\Img(R^{\top})=\Img(R^{\dagger})$~\cite[Exercise
      5.12.16]{CDM:01}. Therefore, for every $x\in
    (\Ker(P))^{\perp}\subset\real^n$, there exists $\alpha\in \real^k$ such
    that $x=R^{\dagger}\alpha$. This implies that $(R^{\dagger}R)x =
    R^{\dagger}R R^{\dagger}\alpha = x$. As a result,
    \begin{align*}
      0 & \ge x^{\top}(PA(R^{\dagger}R)+(R^{\dagger}R)^{\top}A^{\top}P -
      2c P)x  \\ & = x^{\top}(PA+A^{\top}P - 2c P)x. 
    \end{align*}
    Now suppose that inequality~\eqref{eq:semi-demidovich} holds, for
    $P=R^{\top}R$ and for every $x\in \Ker(P)^{\perp}$. As before, for
    every $x\in \Ker(P)^{\perp}$, there exists $\alpha\in \real^k$ such
    that $x=R^{\dagger}\alpha$. Plugging $x =R^{\dagger}\alpha$ into
    inequality~\eqref{eq:semi-demidovich}, we obtain
    \begin{align}\label{eq:usefulinequalitybysaber}
      0 &\ge x^{\top}(PA+A^{\top}P - 2c P)x \nonumber\\
      & = \alpha^{\top} (RAR^{\dagger}+(R^{\dagger})^{\top}A^{\top}R^{\top} -
                                     2c I_{k})\alpha.
    \end{align}
    where the last equality follows from $(R^{\dagger})^{\top}P =
    (R^{\dagger})^{\top}R^{\top}R=(RR^{\dagger})^{\top}R = R$. Since
    inequality~\eqref{eq:semi-demidovich} holds for every $x\in
    \Ker(P)^{\perp}$, inequality~\eqref{eq:usefulinequalitybysaber} holds
    for every $\alpha\in \real^k$.  The
    inequality~\eqref{eq:usefulinequalitybysaber} and the fact that
    $RAR^{\dagger}+(R^{\dagger})^{\top}A^{\top}R^{\top}-2cI_k$ is symmetric
    together imply
    \begin{align*}
      \tfrac{1}{2}\lambda_{\max}
      \left(RAR^{\dagger}+(R^{\dagger})^{\top}A^{\top}R^{\top}\right)\le c.
    \end{align*}
    Thus $\mu_{2,R}(A) = \mu_2(RAR^{\dagger})\le c$. Regarding
    part~\ref{p6:2-norm-inv}, we first show that if $\Ker R$ is
    invariant under $A$, then $PA=PA(R^{\dagger}R)$. Note that, for every $y\in \real^n$,
      we have $y = R^{\dagger}R y+ (I_n-R^{\dagger}R) y$, where
      $R^{\dagger}R y\in \Img(R^{\dagger})$ and $(I_n-R^{\dagger}R)
      y\in \Ker(R)=\Ker(P)$. Now, note that
      \begin{align*}
        PA y = PA(R^{\dagger}Ry + (I_n-R^{\dagger}R) y) = PAR^{\dagger}Ry,
      \end{align*}
      where the last equality holds because $\Ker(R)$ is invariant
      under $A$ and thus $A(I_n-R^{\dagger}R) y \in \Ker(R)$. As a
      result, we get $PA =PAR^{\dagger}R$. Since $R\in\real^{k\times{n}}$ is full rank, we have
    $(\Ker(P))^{\perp}=\Img(R^{\top})=\Img(R^{\dagger})$~\cite[Exercise
      5.12.16]{CDM:01}. Therefore, for every $\alpha\in \real^n$, we
      have $x=R^{\dagger}R\alpha\in \Ker(P)^{\perp}$ and
      \begin{multline*}
        x^{\top}(PA+A^{\top}P-2cP)x =\\ \alpha^{\top}(PA(R^{\dagger}R)
        +(R^{\dagger}R)^{\top}A^{\top}P - 2c P) =\\  \alpha^{\top}(PA+A^{\top}P-2cP)\alpha,
      \end{multline*}
      where, for the first equality, we used the fact that
      $PR^{\dagger}R=R^{\top}RR^{\dagger}R = R^{\top}R = P$ and, for
      the second equality, we used $PA =PAR^{\dagger}R$. Thus,
      equation~\eqref{eq:semi-demidovich} holds for every
      $x\in \Ker(P)^{\perp}$ if and only if $PA+A^{\top}P-2cP$ is
      negative semi-definite. Using part~\ref{p5:2-norm}, we get
      $\mu_{2,R}(A)\le c$ if and only if $PA+A^{\top}P\preceq 2cP$.

    Regarding
    part~\ref{p3:2-norm}, note that by part~\ref{p2:seminorm}, we have
    $\mu_{2,R}(A)=\mu_2(RAR^{\dagger})$. Moreover, using the formula
    for the $\ell_2$-norm matrix measure,
    \begin{align*}
      \mu_2(RAR^{\dagger}) =\tfrac{1}{2}\max\setdef{\lambda}{\lambda\in \mathrm{spec}(RAR^{\dagger} +
      (R^{\dagger})^{\top}A^{\top}R^{\top})}.
    \end{align*}
    Note that $\Ker(R)$ is invariant under $A$. Therefore, using
    Lemma~\ref{lem:useful}, we get $\mathrm{spec}_{\Ker R^{\perp}}(A^{\top}) =
    \mathrm{spec} (RA^{\top}R^{\dagger})$. Noting the fact that $R$ is full
    rank and $RR^{\dagger}=I_k$ and $RP^{\dagger}= R R^{\dagger}
    (R^{\dagger})^{\top} = (R^{\dagger})^{\top}$ and $PR^{\dagger} =
    R^{\top}$, we get
    \begin{align*}
      \mu_2(RAR^{\dagger}) =\tfrac{1}{2}\max\setdef{\lambda}{\lambda\in
      \mathrm{spec}_{\Ker R^{\perp}}(A+ P^{\dagger}A^{\top}P)}.      \hfill\quad    \qedhere
    \end{align*}
  \end{proof}

\section{Networks of diffusively-coupled
  oscillators}\label{app:diffusive-extended}
{\color{black}In this appendix, we provide an extension of
Theorem~\ref{thm:diffusive} for diffusively-coupled identical
dynamical systems~\eqref{eq:diffusive}.

\begin{theorem}[Networks of diffusively-coupled oscillators\textendash{}extended]\label{thm:diffusive2}
  Consider the network of diffusively-coupled identical
  oscillators~\eqref{eq:diffusive} over a connected weighted undirected
  graph with Laplacian matrix $L$. Let $Q\in \real^{k\times k}$ be an
  invertible matrix and $p\in [1,\infty]$. Suppose there exists a positive
  $c$ such that, for every $(t,x)\in \real_{\ge 0}\times \real^k$,
  \begin{align}\label{eq:sync-cond2}
    \mu_{p,Q}(\jac{f}(t,x)) &\leq \lambda_2(L) - c.
   \end{align}
   Then,
  \begin{enumerate}
  \item\label{sync1-2} system~\eqref{eq:diffusive} is semi-contracting with
    exponential rate $c$ with respect to
    $\|\cdot\|_{(2,p),R_{\mcV}\otimes Q}$,
    \item\label{sync3-2} if $t\mapsto x(t)$ and $t\mapsto y(t)$ are
      trajectories of~\eqref{eq:diffusive} starting from
      $x_0\in \real^{nk}$ and $y_0\in \real^{nk}$, respectively, then for every $t\ge 0$, 
      \begin{align*}
        \|y(t) - x(t)\|_{(2,p),R_{\mcV}\otimes Q} \le e^{-ct} \|y_0 - x_0\|_{(2,p),R_{\mcV}\otimes Q} 
        \end{align*}
   \item\label{sync2-2} system~\eqref{eq:diffusive} achieves global
     exponential synchronization with rate $c$, more specifically, for each trajectory
     $x(t) = (x_1(t),\dots,x_n(t))$, we have
     \begin{align*}
       \lim_{t\to \infty} \|x(t) - \vect{1}_n\otimes
       x_{\mathrm{ave}}(t)\|_2 = 0,
     \end{align*}
     where $x_{\mathrm{ave}}(t) = \frac{1}{n}\sum_{i=1}^{n}x_i(t)$.
  \end{enumerate}
\end{theorem}
\begin{proof}
  Set $x = (x_1^{\top},\ldots,x_n^{\top})^{\top}\in \real^{nk}$ and
  define $\mathcal{P}\in \real^{(n+1)\times n}$ by
  \begin{align*}
    \mathcal{P} = \begin{pmatrix}\tfrac{n-1}{n}
      &-\tfrac{1}{n}&\ldots&-\tfrac{1}{n}\\-\tfrac{1}{n}
      &\tfrac{n-1}{n}&\ldots&-\tfrac{1}{n}\\ \vdots & \vdots & \ddots
      & \vdots \\-\tfrac{1}{n}
      &-\tfrac{1}{n}&\ldots&-\tfrac{1}{n} \end{pmatrix}
    \end{align*}
    It is easy to see that $\mathcal{P}^{\dagger} = [I_n,-\vect{1}_n]$. Now
    we define the new variable $z := (\mathcal{P}\otimes I_k) x \in \real^{(n+1)k}$. Since $\mathcal{P}^{\dagger}\mathcal{P} = I_n$,
    we get $x = (\mathcal{P}^{\dagger}\otimes I_k)z$. As a result, in the new
    coordinate chart, the dynamical system~\eqref{eq:diffusive} can be
    written as
    \begin{align}\label{eq:diffusive-matrix}
    \dot{z}= (\mathcal{P}\otimes I_k)F(t,(\mathcal{P}^{\dagger}\otimes I_k)z) - (\mathcal{P} L \mathcal{P}^{\dagger}\otimes I_k)z := G(t,z).
  \end{align}
  where $F(t,x) = [f^{\top}(t,x_1),\ldots,
  f^{\top}(t,x_n)]^{\top}$. It is clear that
  \begin{multline}\label{eq:DG}
    \jac G(t,z) = (\mathcal{P}\otimes I_k)\jac
    F(t,(\mathcal{P}^{\dagger}\otimes
    I_k)z)(\mathcal{P}^{\dagger}\otimes I_k)\\ + (\mathcal{P} L \mathcal{P}^{\dagger}\otimes I_k),
  \end{multline}
  where
  \begin{align*}
    \jac F(t,x) = \begin{bmatrix} \jac f(t,x_1) & 
      \ldots & \vect{0}_{k\times k}\\ \vdots & \ddots &
      \vdots \\ \vect{0}_{k\times k} &
      \ldots & \jac f(t,x_n) 
    \end{bmatrix}\in \real^{nk\times nk}.
  \end{align*}
   Note that $(\mathcal{P}^{\dagger}\otimes I_k)(\vect{1}_n\otimes u) =
\vect{0}_{nk}$, for every $u\in \real^k$. Thus, it is easy to check that
  \begin{align*}
    \Ker\|\cdot\|_{(2,p),R_{\mcV}\mathcal{P}^{\dagger}\otimes Q} = \Ker (R_{\mcV}\mathcal{P}^{\dagger}\otimes Q) =  \mcS
  \end{align*}
  where $\mcS= \mathrm{span}\{\vect{1}_n\otimes u\mid u\in
  \real^k\}$. Moreover, from equation~\eqref{eq:DG}, it is clear that $
  \jac G(t,z) \mcS = \vect{0}_{nk}$, for every $(t,z)\in \real_{\ge
    0}\times \real^{n+1}$. This means that $\mcS$ is
  infinitesimally invariant under the dynamical
  system~\eqref{eq:diffusive-matrix}. Now, we show that the dynamical
  system~\eqref{eq:diffusive-matrix} is semi-contracting with respect
  to the semi-norm $\|\cdot\|_{(2,p),R_{\mcV}\mathcal{P}^{\dagger}\otimes
    Q}$. Note that,
  \begin{multline*}
    \mu_{(2,p), R_{\mcV}\mathcal{P}^{\dagger}\otimes Q} (\jac G(t,x) -
   \mathcal{P} L \mathcal{P}^{\dagger}\otimes I_k)  \\= \mu_{(2,p), R_{\mcV}\otimes Q} (\jac
    F(t,x) - L\otimes I_k) \le -c,
  \end{multline*}
  where the last inequality follows from Theorem~\ref{thm:diffusive}\ref{sync1}. 
Therefore, the dynamical system~\eqref{eq:diffusive-matrix} is semi-contracting with respect
to $\|\cdot\|_{(2,p), R_{\mcV}\mathcal{P}^{\dagger}\otimes Q}$ with rate $c$ and the subspace $\mcS$ is infinitesimally invariant 
under~\eqref{eq:diffusive-matrix}. Moreover, we have $\mathcal{P}^{\dagger}\mathcal{P} = I_n$. This
  implies that, for every $t\ge 0$,
\begin{multline*}
  \|x(t)-y(t)\|_{(2,p), R_{\mcV}\otimes Q}\\=\|(\mathcal{P}\otimes I_k)(x(t)-y(t))\|_{(2,p), R_{\mcV}\mathcal{P}^{\dagger}\otimes Q}.
  \end{multline*}
Thus, using Theorem~\ref{thm:partialcontraction}\ref{p1:estimates}, 
\begin{multline*}
  \|x(t)-y(t)\|_{(2,p), R_{\mcV}\otimes Q}\\ =\|(\mathcal{P}\otimes
                                            I_k)(x(t)-y(t))\|_{(2,p),
                                            R_{\mcV}\mathcal{P}^{\dagger}\otimes
                                            Q} \\ \le e^{-ct}
                                                    \|(\mathcal{P}\otimes
                                                    I_k)(x_0-y_0)\|_{(2,p),
                                                    R_{\mcV}\mathcal{P}^{\dagger}\otimes
                                                    Q} \\ =
  e^{-ct}\|x_0-y_0\|_{(2,p), R_{\mcV}\otimes Q}.
\end{multline*}
This completes the proof of part~\ref{sync3-2}.

Regarding part~\ref{sync2-2}, define $t\mapsto y(t)$ by the following
differential equation
     \begin{align*}
       \dot{y}(t) &= f(t,y(t)), \\
       y(0) &=\tfrac{1}{n}\sum_{i=1}^{n}x_i(0). 
       \end{align*}
Then it is clear that $t\mapsto \vect{1}_n\otimes y(t)$ is a solution
of the dynamical system~\eqref{eq:diffusive}. By part~\ref{sync3-2}, 
\begin{multline}\label{eq:contractingtraj}
  \|x(t)\|_{(2,p),R_{\mcV}\otimes Q} =\|x(t) - \vect{1}_n\otimes y(t)\|_{(2,p),R_{\mcV}\otimes Q} \\ \le
  e^{-ct} \left\|x_0-\vect{1}_n\otimes \tfrac{1}{n}\sum_{i=1}^{n}x_i(0)\right\|_{(2,p),R_{\mcV}\otimes Q}, 
\end{multline}
where the first equality holds because $ \vect{1}_n\otimes y(t)\in
\mcS$. Now, define $P = I_{n} - \tfrac{1}{n}\vect{1}_n\vect{1}_n^{\top}$ and
note that $P=P^\top$ is the orthogonal projection onto
$\vect{1}^{\perp}_n$ and $\Ker P = \mathrm{span}(\vect{1}_n)$. Note
that $P\otimes I_k$ is the orthogonal projection onto
$\mcS^{\perp}$. Thus, equation~\eqref{eq:contractingtraj} implies that
$\lim_{t\to\infty}\|(P\otimes I_k)x(t)\|_2 = 0$ with rate
$c$. Moreover, we know that
\begin{align*}
  (P\otimes I_k)x(t) & = x(t) - \tfrac{1}{n}(\vect{1}_n\vect{1}_n^{\top}\otimes
I_k) x(t) \\ & = x(t) - (\vect{1}_n\otimes \subscr{x}{ave}(t)).
\end{align*}
As a result, $\lim_{t\to\infty}\|x(t) - (\vect{1}_n\otimes \subscr{x}{ave}(t))\|_2=0$
with rate $c$, which in turn means that
$\lim_{t\to\infty}\|x_i(t)-x_j(t)\|_2=0$ with rate $c$, for every $i,j\in \{1,\ldots,n\}$.
\end{proof}}

\section{The Lotka\textendash{}Volterra model}\label{sec:LV}

In this appendix, we study the Lotka\textendash{}Volterra model
with mutualistic interactions and prove that it is weakly
contracting. The Lotka\textendash{}Volterra model is given by
\begin{align}\label{eq:LV}
  \dot{x} = \diag(x)(Ax + r),
  \end{align}
where $x=(x_1,\ldots,x_n)^{\top}\in\real_{\ge 0}^n$ is the vector of
populations, $A\in \real^{n\times n}$ is the mutual interaction matrix, and
$r>\vect{0}_n$ is the intrinsic growth rate. We assume $A$ is Metzler,
i.e., the interactions between any two species are mutualistic.

\begin{theorem}[Global asymptotic stability and Lyapunov functions]\label{thm:saberLV}
  Consider the Lotka\textendash{}Volterra model~\eqref{eq:LV} with a
  Metzler and Hurwitz interaction matrix $A$. Let the positive vector
  $v\in\realpositive^n$ satisfy $v^\top{A}<\vect{0}_n$. Then
  \begin{enumerate}
  \item\label{fact:saberLV:1} the open positive orthant $\realpositive^n$
    is invariant,
  \item\label{fact:saberLV:2} $x^* = -A^{-1}r> \vect{0}_{n}$ is the unique
    globally asymptotically stable equilibrium point of~\eqref{eq:LV}
    restricted to $\realpositive^n$, 
  \item\label{fact:saberLV:3} the following distance between any two
    trajectories $x(t)$ and $z(t)$ is decreasing:
    \begin{equation*}
      d_{\mathrm{LV}}(x(t),z(t))=\sum\nolimits_{i=1}^{n} v_i |\ln\big(x_i(t)/z_i(t)\big)|,
    \end{equation*}
  \item\label{fact:saberLV:4} the following functions are global Lyapunov
    functions:
    \begin{align*}
      x\mapsto\sum\nolimits_{i=1}^{n} v_i |\ln(x_i/x^*_i)|,\quad x\mapsto\sum\nolimits_{i=1}^{n} v_i|(Ax+r)_i|.
    \end{align*}
  \end{enumerate}
\end{theorem}

\begin{proof}
  We omit the proof of statement~\ref{fact:saberLV:1} in the interest of
  brevity.  For $x\in\realpositive^n$, let $y_i = \ln(x_i)\in\real$, $i\in
  \until{n}$ and write the Lotka\textendash{}Volterra model~\eqref{eq:LV}
  as
  \begin{align}\label{eq:LV-c}
    \dot{y} = A\exp(y) + r := \subscr{f}{LVe}(y),
  \end{align}
  where $y$ and its entry-wise exponential $\exp(y)$ are vectors in
  $\real^{n}$.  Note that $\jac \fLVe (y) = A \diag(\exp(y))$ is Metzler
  since $\exp(y)>\vect{0}_n$.  Since $v\in \real^n_{>0}$ satisfies $v^\top
  A < \vect{0}_n$, there exists $c>0$ such that $v^{\top}A \leq -c
  v^{\top}$. Therefore, for every $y\in\real^n$, 
  \begin{align*}
    v^{\top}\jac \fLVe (y) = v^{\top} A \diag(\exp(y)) < -c
    v^{\top}\diag(\exp(y)) \le 0.
  \end{align*}
  We now recall~\cite{SC:19} that, for a Metzler matrix $M$,
  a positive vector $v$ and a scalar $b$,
  \begin{equation*} 
    v^\top M \leq bv^\top \enspace\iff\enspace 
    \mu_{1,\diag(v)^{-1}}(M) \leq b .
  \end{equation*}
  This equivalence implies that $\subscr{f}{LVe}$ is
  weakly contracting in its domain.  After a change of coordinates, this
  establishes statement~\ref{fact:saberLV:3}.  At the equilibrium point
  $x^*=-A^{-1}r>\vect{0}_n$, that is, at $\exp(y^*)= -A^{-1}r
  \in\realpositive^n$, we know
  \begin{align*}
    v^{\top}\jac \fLVe(y^*) &\leq -c v^{\top} \diag(\exp(y^*)) \\
    &=  -c v^{\top} \diag(-A^{-1}r)  \\
    & \leq      - c \Big(\min_{i\in\until{n}}(-A^{-1}r)_i\Big)   v^{\top},
    \\
     \implies & \mu_{1,\diag(v)}\big(\jac \fLVe(y^*)\big) \leq -
     c \min_{i\in\until{n}}(-A^{-1}r)_i .
  \end{align*}
  We now invoke Theorem~\ref{thm:weak-dichotomy}\ref{thmwd:p4.5}: $\fLVe$ is
  weakly contracting over the entire $\real^n$ and has strictly negative
  matrix measure at the equilibrium  $y^*$.  Therefore, $y^*$ is the
  unique globally exponentially stable equilibrium with global Lyapunov
  functions
  \begin{equation*}
    y\mapsto  \norm{y-y^*}_{1,\diag(v)},
    \enspace\text{ and }\enspace
    y\mapsto  \norm{\fLVe(y)}_{1,\diag(v)}.
  \end{equation*}
  Facts~\ref{fact:saberLV:2} and~\ref{fact:saberLV:4} follow from a
  change of coordinates.
\end{proof}

\bibliographystyle{plainurl+isbn}
\bibliography{alias,Main,FB,New}

\begin{thebibliography}{10}

\bibitem{ZA-EDS:14b}
Z.~Aminzare and E.~D. Sontag.
\newblock Contraction methods for nonlinear systems: {A} brief introduction and
  some open problems.
\newblock In {\em {IEEE} Conf.\ on Decision and Control}, pages 3835--3847,
  December 2014.
\newblock \href {http://dx.doi.org/10.1109/CDC.2014.7039986}
  {\path{doi:10.1109/CDC.2014.7039986}}.

\bibitem{ZA-EDS:14}
Z.~Aminzare and E.~D. Sontag.
\newblock Synchronization of diffusively-connected nonlinear systems: {R}esults
  based on contractions with respect to general norms.
\newblock {\em IEEE Transactions on Network Science and Engineering},
  1(2):91--106, 2014.
\newblock \href {http://dx.doi.org/10.1109/TNSE.2015.2395075}
  {\path{doi:10.1109/TNSE.2015.2395075}}.

\bibitem{MA:11}
M.~Arcak.
\newblock Certifying spatially uniform behavior in reaction-diffusion {PDE} and
  compartmental {ODE} systems.
\newblock {\em Automatica}, 47(6):1219--1229, 2011.
\newblock \href {http://dx.doi.org/10.1016/j.automatica.2011.01.010}
  {\path{doi:10.1016/j.automatica.2011.01.010}}.

\bibitem{EA-PAP-JJES:08}
E.~M. Aylward, P.~A. Parrilo, and J-J.~E. Slotine.
\newblock Stability and robustness analysis of nonlinear systems via
  contraction metrics and {SOS} programming.
\newblock {\em Automatica}, 44(8):2163--2170, 2008.
\newblock \href {http://dx.doi.org/10.1016/j.automatica.2007.12.012}
  {\path{doi:10.1016/j.automatica.2007.12.012}}.

\bibitem{SB-LV:04}
S.~Boyd and L.~Vandenberghe.
\newblock {\em Convex Optimization}.
\newblock Cambridge University Press, 2004, ISBN 0521833787.

\bibitem{FB:20}
F.~Bullo.
\newblock {\em Lectures on Network Systems}.
\newblock Kindle Direct Publishing, {1.4} edition, July 2020, ISBN
  978-1986425643.
\newblock With contributions by J. Cort{\'e}s, F. D\"orfler, and S.
  Mart{\'\i}nez.
\newblock URL: \url{http://motion.me.ucsb.edu/book-lns}.

\bibitem{MC:07a}
M.~Chen.
\newblock Synchronization in time-varying networks: {A} matrix measure
  approach.
\newblock {\em Physical Review E}, 76:016104, 2007.
\newblock \href {http://dx.doi.org/10.1103/PhysRevE.76.016104}
  {\path{doi:10.1103/PhysRevE.76.016104}}.

\bibitem{AC-BG-JC:17}
A.~Cherukuri, B.~Gharesifard, and J.~Cortes.
\newblock Saddle-point dynamics: {Conditions} for asymptotic stability of
  saddle points.
\newblock {\em SIAM Journal on Control and Optimization}, 55(1):486--511, 2017.
\newblock \href {http://dx.doi.org/10.1137/15M1026924}
  {\path{doi:10.1137/15M1026924}}.

\bibitem{LOC-LY:88}
L.~O. Chua and L.~Yang.
\newblock Cellular neural networks: Theory.
\newblock {\em IEEE Transactions on Circuits and Systems}, 35(10):1257--1272,
  1988.
\newblock \href {http://dx.doi.org/10.1109/31.7600}
  {\path{doi:10.1109/31.7600}}.

\bibitem{PCV-SJ-FB:19r}
P.~Cisneros-Velarde, S.~Jafarpour, and F.~Bullo.
\newblock Distributed and time-varying primal-dual dynamics via contraction
  analysis.
\newblock {\em IEEE Transactions on Automatic Control}, March 2020.
\newblock Submitted.
\newblock URL: \url{https://arxiv.org/pdf/2003.12665}.

\bibitem{GC:17}
G.~Como.
\newblock On resilient control of dynamical flow networks.
\newblock {\em Annual Reviews in Control}, 43:80--90, 2017.
\newblock \href {http://dx.doi.org/10.1016/j.arcontrol.2017.01.001}
  {\path{doi:10.1016/j.arcontrol.2017.01.001}}.

\bibitem{SC:19}
S.~Coogan.
\newblock A contractive approach to separable {Lyapunov} functions for monotone
  systems.
\newblock {\em Automatica}, 106:349--357, 2019.
\newblock \href {http://dx.doi.org/10.1016/j.automatica.2019.05.001}
  {\path{doi:10.1016/j.automatica.2019.05.001}}.

\bibitem{WAC:1965}
W.~A. Coppel.
\newblock {\em Stability and Asymptotic Behavior Of Differential Equations}.
\newblock Heath, 1965, ISBN 0669190187.

\bibitem{PD-MdB-GR:11}
P.~DeLellis, M.~{Di~Bernardo}, and G.~Russo.
\newblock On {QUAD}, {Lipschitz}, and contracting vector fields for consensus
  and synchronization of networks.
\newblock {\em IEEE Transactions on Circuits and Systems I: Regular Papers},
  58(3):576--583, 2011.
\newblock \href {http://dx.doi.org/10.1109/TCSI.2010.2072270}
  {\path{doi:10.1109/TCSI.2010.2072270}}.

\bibitem{BPD:61a}
B.~P. Demidovi\v{c}.
\newblock On the dissipativity of a certain non-linear system of differential
  equations. {I}.
\newblock {\em Vestnik Moskovskogo Universiteta. Serija I. Matematika,
  Mehanika}, 6:19--27, 1961.

\bibitem{CAD-MV:1975}
C.~A. Desoer and M.~Vidyasagar.
\newblock {\em Feedback Systems: Input-Output Properties}.
\newblock Academic Press, 1975, ISBN 978-0-12-212050-3.
\newblock \href {http://dx.doi.org/10.1137/1.9780898719055}
  {\path{doi:10.1137/1.9780898719055}}.

\bibitem{MdB-DF-GR-FS:16}
M.~{Di~Bernardo}, D.~Fiore, G.~Russo, and F.~Scafuti.
\newblock Convergence, consensus and synchronization of complex networks via
  contraction theory.
\newblock In J.~L{\"u}, X.~Yu, G.~Chen, and W.~Yu, editors, {\em Complex
  Systems and Networks: Dynamics, Controls and Applications}, pages 313--339.
  Springer, 2016.
\newblock \href {http://dx.doi.org/10.1007/978-3-662-47824-0_12}
  {\path{doi:10.1007/978-3-662-47824-0_12}}.

\bibitem{DD-MRJ:18}
D.~{Ding} and M.~R. {Jovanović}.
\newblock A primal-dual {Laplacian} gradient flow dynamics for distributed
  resource allocation problems.
\newblock In {\em {A}merican {C}ontrol {C}onference}, pages 5316--5320, 2018.
\newblock \href {http://dx.doi.org/10.23919/ACC.2018.8431779}
  {\path{doi:10.23919/ACC.2018.8431779}}.

\bibitem{DF-FP:10}
D.~Feijer and F.~Paganini.
\newblock Stability of primal--dual gradient dynamics and applications to
  network optimization.
\newblock {\em Automatica}, 46(12):1974--1981, 2010.
\newblock \href {http://dx.doi.org/10.1016/j.automatica.2010.08.011}
  {\path{doi:10.1016/j.automatica.2010.08.011}}.

\bibitem{RFH:61}
R.~FitzHugh.
\newblock Impulses and physiological states in theoretical models of nerve
  membrane.
\newblock {\em Biophysical Journal}, 1(6):445--466, 1961.
\newblock \href {http://dx.doi.org/10.1016/S0006-3495(61)86902-6}
  {\path{doi:10.1016/S0006-3495(61)86902-6}}.

\bibitem{FF-RS:14}
F.~Forni and R.~Sepulchre.
\newblock A differential {Lyapunov} framework for contraction analysis.
\newblock {\em IEEE Transactions on Automatic Control}, 59(3):614--628, 2014.
\newblock \href {http://dx.doi.org/10.1109/TAC.2013.2285771}
  {\path{doi:10.1109/TAC.2013.2285771}}.

\bibitem{FF-RS:16}
F.~{Forni} and R.~{Sepulchre}.
\newblock Differentially positive systems.
\newblock {\em IEEE Transactions on Automatic Control}, 61(2):346--359, 2016.
\newblock \href {http://dx.doi.org/10.1109/TAC.2015.2437523}
  {\path{doi:10.1109/TAC.2015.2437523}}.

\bibitem{BCG:65}
B.~C. Goodwin.
\newblock Oscillatory behavior in enzymatic control processes.
\newblock {\em Advances in Enzyme Regulation}, 3:425--437, 1965.
\newblock \href {http://dx.doi.org/10.1016/0065-2571(65)90067-1}
  {\path{doi:10.1016/0065-2571(65)90067-1}}.

\bibitem{RAH-CRJ:12}
R.~A. Horn and C.~R. Johnson.
\newblock {\em Matrix Analysis}.
\newblock Cambridge University Press, 2nd edition, 2012, ISBN 0521548233.

\bibitem{YK-BB-MC:20}
Y.~Kawano, B.~Besselink, and M.~Cao.
\newblock Contraction analysis of monotone systems via separable functions.
\newblock {\em IEEE Transactions on Automatic Control}, 65(8):3486--3501, 2020.
\newblock \href {http://dx.doi.org/10.1109/TAC.2019.2944923}
  {\path{doi:10.1109/TAC.2019.2944923}}.

\bibitem{VVK:83}
V.~V. Kolpakov.
\newblock Matrix seminorms and related inequalities.
\newblock {\em Journal of Soviet Mathematics}, 23:2094–2106, 1983.
\newblock \href {http://dx.doi.org/10.1007/BF01093289}
  {\path{doi:10.1007/BF01093289}}.

\bibitem{AL-JAY:76}
A.~Lajmanovich and J.~A. Yorke.
\newblock A deterministic model for gonorrhea in a nonhomogeneous population.
\newblock {\em Mathematical Biosciences}, 28(3):221--236, 1976.
\newblock \href {http://dx.doi.org/10.1016/0025-5564(76)90125-5}
  {\path{doi:10.1016/0025-5564(76)90125-5}}.

\bibitem{DCL:49}
D.~C. Lewis.
\newblock Metric properties of differential equations.
\newblock {\em American Journal of Mathematics}, 71(2):294--312, 1949.
\newblock \href {http://dx.doi.org/10.2307/2372245}
  {\path{doi:10.2307/2372245}}.

\bibitem{MYL-LW:98}
M.~Y. Li and L.~Wang.
\newblock A criterion for stability of matrices.
\newblock {\em Journal of Mathematical Analysis and Applications},
  225(1):249--264, 1998.
\newblock \href {http://dx.doi.org/10.1006/jmaa.1998.6020}
  {\path{doi:10.1006/jmaa.1998.6020}}.

\bibitem{NL-CL-ZC-SHL:13}
Na~Li, Lijun Chen, Changhong Zhao, and Steven~H. Low.
\newblock Connecting automatic generation control and economic dispatch from an
  optimization view.
\newblock In {\em {A}merican {C}ontrol {C}onference}, pages 735--740, Portland,
  OR, USA, June 2014.

\bibitem{WL-JJES:98}
W.~Lohmiller and J.-J.~E. Slotine.
\newblock On contraction analysis for non-linear systems.
\newblock {\em Automatica}, 34(6):683--696, 1998.
\newblock \href {http://dx.doi.org/10.1016/S0005-1098(98)00019-3}
  {\path{doi:10.1016/S0005-1098(98)00019-3}}.

\bibitem{EL-GC-KS:14}
E.~{Lovisari}, G.~{Como}, and K.~{Savla}.
\newblock Stability of monotone dynamical flow networks.
\newblock In {\em {IEEE} Conf.\ on Decision and Control}, pages 2384--2389, Los
  Angeles, USA, December 2014.
\newblock \href {http://dx.doi.org/10.1109/CDC.2014.7039752}
  {\path{doi:10.1109/CDC.2014.7039752}}.

\bibitem{WL-TC:06}
W.~Lu and T.~Chen.
\newblock New approach to synchronization analysis of linearly coupled ordinary
  differential systems.
\newblock {\em Physica D: Nonlinear Phenomena}, 213(2):214--230, 2006.
\newblock \href {http://dx.doi.org/10.1016/j.physd.2005.11.009}
  {\path{doi:10.1016/j.physd.2005.11.009}}.

\bibitem{IRM-JES:14}
I.~R. Manchester and J.-J.~E. Slotine.
\newblock Transverse contraction criteria for existence, stability, and
  robustness of a limit cycle.
\newblock {\em Systems \& Control Letters}, 63:32--38, 2014.
\newblock \href {http://dx.doi.org/10.1016/j.sysconle.2013.10.005}
  {\path{doi:10.1016/j.sysconle.2013.10.005}}.

\bibitem{MM-EDS-TT:16}
M.~Margaliot, E.~D. Sontag, and T.~Tuller.
\newblock Contraction after small transients.
\newblock {\em Automatica}, 67:178--184, 2016.
\newblock \href {http://dx.doi.org/10.1016/j.automatica.2016.01.018}
  {\path{doi:10.1016/j.automatica.2016.01.018}}.

\bibitem{CDM:01}
C.~D. Meyer.
\newblock {\em Matrix Analysis and Applied Linear Algebra}.
\newblock SIAM, 2001, ISBN 0898714540.

\bibitem{JM:00}
J.~Munkres.
\newblock {\em Topology}.
\newblock Pearson, 2 edition, 2000, ISBN 9780131816299.

\bibitem{HDN-TLV-KT-JJS:18}
H.~D. {Nguyen}, T.~L. {Vu}, K.~{Turitsyn}, and J.~{Slotine}.
\newblock Contraction and robustness of continuous time primal-dual dynamics.
\newblock {\em IEEE Control Systems Letters}, 2(4):755--760, 2018.
\newblock \href {http://dx.doi.org/10.1109/LCSYS.2018.2847408}
  {\path{doi:10.1109/LCSYS.2018.2847408}}.

\bibitem{AP-AP-NVDW-HN:04}
A.~Pavlov, A.~Pogromsky, N.~Van~de Wouw, and H.~Nijmeijer.
\newblock Convergent dynamics, a tribute to {B}oris {P}avlovich {D}emidovich.
\newblock {\em Systems \& Control Letters}, 52(3-4):257--261, 2004.
\newblock \href {http://dx.doi.org/10.1016/j.sysconle.2004.02.003}
  {\path{doi:10.1016/j.sysconle.2004.02.003}}.

\bibitem{QCP-JJS:07}
Q.~C. Pham and J.~J. Slotine.
\newblock Stable concurrent synchronization in dynamic system networks.
\newblock {\em Neural Networks}, 20(1):62--77, 2007.
\newblock \href {http://dx.doi.org/10.1016/j.neunet.2006.07.008}
  {\path{doi:10.1016/j.neunet.2006.07.008}}.

\bibitem{GQ-NL:19}
G.~{Qu} and N.~{Li}.
\newblock On the exponential stability of primal-dual gradient dynamics.
\newblock {\em IEEE Control Systems Letters}, 3(1):43--48, 2019.
\newblock \href {http://dx.doi.org/10.1109/LCSYS.2018.2851375}
  {\path{doi:10.1109/LCSYS.2018.2851375}}.

\bibitem{GR-MDB-EDS:10a}
G.~Russo, M.~{Di~Bernardo}, and E.~D. Sontag.
\newblock Global entrainment of transcriptional systems to periodic inputs.
\newblock {\em PLOS Computational Biology}, 6(4):e1000739, 2010.
\newblock \href {http://dx.doi.org/10.1371/journal.pcbi.1000739}
  {\path{doi:10.1371/journal.pcbi.1000739}}.

\bibitem{GR-MDB-EDS:13}
G.~{Russo}, M.~{Di~Bernardo}, and E.~D. {Sontag}.
\newblock A contraction approach to the hierarchical analysis and design of
  networked systems.
\newblock {\em IEEE Transactions on Automatic Control}, 58(5):1328--1331, 2013.
\newblock \href {http://dx.doi.org/10.1109/TAC.2012.2223355}
  {\path{doi:10.1109/TAC.2012.2223355}}.

\bibitem{RAR:13}
R.~A. Ryan.
\newblock {\em Introduction to Tensor Products of Banach Spaces}.
\newblock Springer, 2002, ISBN 9781852334376.

\bibitem{LS-RS:09}
L.~Scardovi and R.~Sepulchre.
\newblock Synchronization in networks of identical linear systems.
\newblock {\em Automatica}, 45(11):2557--2562, 2009.
\newblock \href {http://dx.doi.org/10.1016/j.automatica.2009.07.006}
  {\path{doi:10.1016/j.automatica.2009.07.006}}.

\bibitem{JWSP-FB:12za}
J.~W. Simpson-Porco and F.~Bullo.
\newblock Contraction theory on {R}iemannian manifolds.
\newblock {\em Systems \& Control Letters}, 65:74--80, 2014.
\newblock \href {http://dx.doi.org/10.1016/j.sysconle.2013.12.016}
  {\path{doi:10.1016/j.sysconle.2013.12.016}}.

\bibitem{JJS:03}
J.-J. Slotine.
\newblock Modular stability tools for distributed computation and control.
\newblock {\em International Journal of Adaptive Control and Signal
  Processing}, 17(6):397--416, 2003.
\newblock \href {http://dx.doi.org/10.1002/acs.754}
  {\path{doi:10.1002/acs.754}}.

\bibitem{AMT:52}
A.~M. Turing.
\newblock The chemical basis of morphogenesis.
\newblock {\em Philosophical Transactions of the Royal Society of London.
  Series B, Biological Sciences}, 237(641):37--72, 1952.
\newblock \href {http://dx.doi.org/10.1098/rstb.1952.0012}
  {\path{doi:10.1098/rstb.1952.0012}}.

\bibitem{MV:78-book}
M.~Vidyasagar.
\newblock {\em Nonlinear Systems Analysis}.
\newblock Prentice Hall, 1978, ISBN 0136232809.
\newblock \href {http://dx.doi.org/10.1137/1.9780898719185}
  {\path{doi:10.1137/1.9780898719185}}.

\bibitem{MV:78}
M.~Vidyasagar.
\newblock On matrix measures and convex {Liapunov} functions.
\newblock {\em Journal of Mathematical Analysis and Applications},
  62(1):90--103, 1978.
\newblock \href {http://dx.doi.org/10.1016/0022-247X(78)90221-4}
  {\path{doi:10.1016/0022-247X(78)90221-4}}.

\bibitem{JW-NE:11}
J.~Wang and N.~Elia.
\newblock A control perspective for centralized and distributed convex
  optimization.
\newblock In {\em {IEEE} Conf.\ on Decision and Control and European Control
  Conference}, pages 3800--3805, Orlando, USA, 2011.
\newblock \href {http://dx.doi.org/10.1109/CDC.2011.6161503}
  {\path{doi:10.1109/CDC.2011.6161503}}.

\bibitem{WW-JJES:05}
W.~Wang and J.-J.~E. Slotine.
\newblock On partial contraction analysis for coupled nonlinear oscillators.
\newblock {\em Biological Cybernetics}, 92(1):38--53, 2005.
\newblock \href {http://dx.doi.org/10.1007/s00422-004-0527-x}
  {\path{doi:10.1007/s00422-004-0527-x}}.

\bibitem{CWW-LOC:95b}
C.~W. Wu and L.~O. Chua.
\newblock Synchronization in an array of linearly coupled dynamical systems.
\newblock {\em IEEE Transactions on Circuits and Systems~I: Fundamental Theory
  and Applications}, 42:430--447, 1995.
\newblock \href {http://dx.doi.org/10.1109/81.404047}
  {\path{doi:10.1109/81.404047}}.

\bibitem{DW:20}
D.~Wu.
\newblock On geometric and {Lyapunov} characterizations of incremental stable
  systems on {Finsler} manifolds, 2020.
\newblock Arxiv e-print.
\newblock URL: \url{https://arxiv.org/pdf/2002.11444v1}.

\bibitem{TY-XY-JW-DW:19}
T.~Yang, X.~Yi, J.~Wu, Y.~Yuan, D.~Wu, Z.~Meng, Y.~Hong, H.~Wang, Z.~Lin, and
  K.~H. Johansson.
\newblock A survey of distributed optimization.
\newblock {\em Annual Reviews in Control}, 47:278--305, 2019.
\newblock \href {http://dx.doi.org/10.1016/j.arcontrol.2019.05.006}
  {\path{doi:10.1016/j.arcontrol.2019.05.006}}.

\end{thebibliography}

\end{document}